\documentclass[11pt,journal,draftclsnofoot, onecolumn]{IEEEtran}
\usepackage{url}          % example: \url{my_url_here}.
\usepackage[cmex10]{amsmath}
\usepackage{amsfonts}
\usepackage{amssymb}
\usepackage{cite}
\usepackage[caption=false,font=footnotesize]{subfig}

\ifCLASSINFOpdf
  \usepackage[pdftex]{graphicx}
  \DeclareGraphicsExtensions{.pdf,.jpeg,.jpg,.png}
\else
  \usepackage[dvips]{graphicx}
  \DeclareGraphicsExtensions{.eps}
\fi

\newtheorem{theorem}{Property}%[section]        % 1) [section] 2) remove [section]

\begin{document}

%% can use linebreaks \\ within to get better formatting as desired
\title{Selective Combining for Hybrid Cooperative Networks}
% {Generalized Hybrid Relay Selection for Cooperative Wireless Networks}
% {Generalized Hybrid Relay Selection Scheme for Cooperative Wireless Networks}
%Performance Analysis for Generalized Hybrid Relay Selection Scheme

% -- for IET COMMUN -----

%\author{\IEEEauthorblockN{Qiang~Huo\IEEEauthorrefmark{1},
%Tianxi~Liu\IEEEauthorrefmark{2}, Shaohui~Sun\IEEEauthorrefmark{3},
%Lingyang~Song\IEEEauthorrefmark{1} and
%Bingli~Jiao\IEEEauthorrefmark{1}}
%\\
%\IEEEauthorblockA{\IEEEauthorrefmark{1}
%State Key Laboratory of Advanced Optical Communication Systems and Networks, and\\
%Wireless Communications and Signal Processing Research Center\\
%School of Electronics Engineering and Computer Science\\
%Peking University, Beijing, China, 100871\\
%Email: \{qiang.huo, lingyang.song, jiaobl\}@pku.edu.cn}
%\\ \IEEEauthorblockA{\IEEEauthorrefmark{2}
%Datang Telecom Technology \& Industry Holdings, Beijing, China, 100871\\
%Email: liutx@datanggroup.cn}
%\\
%\IEEEauthorblockA{\IEEEauthorrefmark{3}
%Peking University, Beijing, China, 100871\\
%Email: sunshaohui@hotmail.com}
%}

% -- for arXiv -----

\author{Qiang~Huo,
        Tianxi~Liu,
        Shaohui~Sun,
        Lingyang~Song,
        and~Bingli~Jiao%
\thanks{Published in IET Communications.
Manuscript received  May 5, 2013; revised August 26, 2013; accepted September 23, 2013.}%
\thanks{Q.~Huo,  L.~Song, and B.~Jiao are with the State Key Laboratory of Advanced Optical Communication Systems and Networks, and the
Wireless Communications and Signal Processing Research Center,
School of Electronics Engineering and Computer Science,
Peking University, Beijing 100871, China (e-mail: \{qiang.huo, lingyang.song, jiaobl\}@pku.edu.cn).}%
\thanks{T.~Liu is with Datang Telecom Technology \& Industry Holdings, Beijing 100871, China (e-mail: Email: liutx@datanggroup.cn).}%
\thanks{S.~Sun is with Peking University, Beijing 100871, China (e-mail: sunshaohui@hotmail.com).}%
\thanks{This work was partially supported by the National 973 project under grant number 2013CB336700,
National Nature Science Foundation of China under grant number 61222104
and 61061130561, the Ph.D. Programs Foundation of Ministry of Education of
China under grant number 20110001110102, and the Opening Project of Key
Laboratory of Cognitive Radio and Information Processing (Guilin University
of Electronic Technology). }%
\thanks{Digital Object Identifier 10.1049/iet-com.2013.0323}%
}

\maketitle
%% ----------------------

%% paper headers and etc.
%% ==================================== %%

%% ==================================== %%
%% abstract section

\begin{abstract}
%\boldmath
In this study, we consider the selective combining in hybrid cooperative networks~(SCHCNs scheme) with one source node, one destination node and $N$ relay nodes.
In the SCHCN scheme, each relay first adaptively chooses  between amplify-and-forward   protocol and decode-and-forward   protocol on a per frame basis by examining the error-detecting code result, and $N_c$ ($1\leq N_c \leq N$) relays will be selected to forward their received signals to the destination. We first develop a signal-to-noise ratio~(SNR) threshold-based frame error rate~(FER) approximation model. Then, the theoretical FER expressions for the SCHCN scheme are derived by utilizing the proposed SNR threshold-based FER approximation model.
The analytical FER expressions are validated through simulation results.
\end{abstract}

%% no keywords in IEEE conference model

\IEEEpeerreviewmaketitle

%% abstract section
%% ==================================== %%

%% ----------------------
%\newpage
%% ----------------------

%% ==================================== %%
%% Introduction section

\section{Introduction}\label{sec:introduction}
% \IEEEPARstart{A}{demo} file
% \IEEEPARstart{T}{his} demo file
%% no IEEEPARstart in conference model
\IEEEPARstart{C}{ooperative} communication has been widely used to improve system robustness and capacity by allowing nodes to cooperate in their transmissions to form a virtual antenna array~\cite{ Wang2012P289, Cheng2012P388, Sendonaris2003P1927,Sendonaris2003P1939,Laneman2004P3062}.
In \cite{Laneman2004P3062}, several relay protocols have been studied for the wireless cooperative networks, and amplify-and-forward (AF)  and decode-and-forward (DF) are recognized as two basic cooperative transmission protocols.
While using the AF protocol, the relay node amplifies the received signal sent by the source, and then retransmits it to the destination without decoding. In contrast, the DF protocol requires the relay node to decode the received signal and then forward the re-encoded signal to the destination.
However, the performance of the AF  is mainly limited by the noise amplification phenomenon at the relay nodes  %, especially at low signal-to-noise ratio~(SNR)
and the DF  may cause serious error propagation problem when the  relay  fails to decode the received signal correctly\cite{Onat2008P4226,Liu2009P1,Duong2009P1,Duong2009P1a}.

Motivated by the above-described disadvantages of the AF and DF protocols, various improved cooperative transmission protocols have been proposed, such as signal-to-noise ratio (SNR) threshold-based selective digital relaying schemes\cite{Onat2007P969,Onat2008P4226,Onat2008P1,Onat2008P4938}, cooperative maximum ratio combining (MRC)-based DF\cite{Wang2005P1051,Wang2007P1427} and smart relaying scheme\cite{Wang2008P1950,Vien2009P2849},  log-likelihood-ratio~(LLR) threshold-based selective DF\cite{Kwon2010P540}, decode-amplify-forward~(DAF)\cite{Bao2005P816,Bao2007P3975}.
%such as SNR threshold-based selective digital relaying schemes\cite{Onat2007P969,Onat2008P4226,Onat2008P1,Onat2008P4938}, cooperative maximum-ratio-combining-based DF\cite{Wang2005P1051,Wang2007P1427}, link adaptive relaying DF\cite{Wang2008P1950,Vien2009P2849}, decode-amplify-forward~(DAF)\cite{Bao2005P816,Bao2007P3975}, log-likelihood-ratio~(LLR) threshold-based selective DF\cite{Kwon2010P540} and soft information relaying\cite{Li2006P2040}, etc.
A hybrid forward~(HF) protocol was proposed in \cite{Can2006P4520,Li2007P3704,Duong2009P1,Duong2009P1a,Song2009P1,Liu2011P457},  which provides superior performance than both the AF and the DF protocols and thus has received lots of attention  recently.
In \cite{Duong2009P1,Duong2009P1a}, the author studied the symbol error rate~(SER) of the HF protocol by adaptively choosing the AF and the DF protocols on a per symbol basis, which is, however, hard to realise in practice.
This work focuses on a more practical HF protocol on a per frame basis.
%Recently, to overcome the above described disadvantages of the AF and DF protocols, a hybrid forward~(HF) protocol was proposed in \cite{Can2006P4520,Li2007P3704,Duong2009P1,Duong2009P1a,Song2009P1,Liu2011P457},  which provides superior performance than both the AF and the DF protocols obviously.
%In \cite{Duong2009P1}, the author analyzed the symbol error probability~(SEP) of the HF protocol by adaptively selecting the AF and the DF protocols on a per symbol basis, which is, however, hard to realize in practice.

%On the other hand,
In conventional cooperative wireless networks, all $N$ relays will forward the signals received from the source to the destination through $N$
orthogonal channels in the relaying phase\cite{Laneman2003P2415,Anghel2004P1416,Ribeiro2005P1264,Zhao2006P757,Zhao2007P3114}.
Then the destination combines all the   signals received from the source and the relays using  MRC  technique.
In \cite{Zhao2006P757,Zhao2007P3114}, the authors show that the all participate AF~(AP-AF) scheme can achieve the full diversity order of $N+1$. However, MRC, which combines the signals received from all the links in order to have the best possible combining   gain, will certainly result in extremely complicated hardware implementation required for phase coherence and amplitude estimation on each link\cite{Sulyman2000P3,Chen2004P2}. Furthermore, note that the links
with low SNRs may possibly lead to inefficient estimation on both phase and amplitude; hence MRC is sensitive to the channel estimation errors\cite{Li2006P4230}.
To simplify the hardware implementation and make the  communication systems more robust towards channel estimation errors while still retaining good system
performance, Kong \textit{et al.} \cite{Kong1995P426} first proposed the selection combining~(SC) scheme at the receiver, which combines $N_c$ links with the largest instantaneous SNRs out of $N$ links, for a rake receiver operating over a multipath fading channel.
As a result, the SC scheme has a fixed and low processing complexity since it only combines $N_c$ links, instead of all the $N$ links in the network\cite{Chen2004P2}.
In addition, in the SC scheme, the weakest SNR paths are excluded from the combining process to make the system more robust towards channel estimation errors.
% Since $N_c$ is fixed and the receiver only combines $N_c$ links, instead of all the $N$ links in the network,   the SC scheme has a fixed and low processing complexity.

%\pagebreak
Recently, the SC scheme  was extended to the cooperative communication networks by selecting proper number of relay nodes to forward the signals to the destination for the source during the transmission phase.
%such that the receiver may combine all the received signals in place of performing GSC, known as the generalized relay selection~(GRS).
The work in \cite{Ikki2009P1} and \cite{Ikki2010P1} studies the SER performances of the SC with the AF and the DF protocols on a per symbol basis.
Note that the conventional cooperative scheme\cite{Laneman2003P2415,Anghel2004P1416,Ribeiro2005P1264,Zhao2006P757,Zhao2007P3114}, in which all $N$ relays will forward the signals to the destination, can be viewed as a special case of the SC scheme when $N_c=N$, and the conventional relay selection~(RS)
scheme\cite{Beres2006P1056,Yi2008P1792,Li2007P3704,Song2009P1}, in which only the relay that has the highest SNRs at
the destination is chosen as the active relay, is also a special case of the SC scheme when $N_c=1$. Obviously, the SC scheme is able to describe the cooperative networks in a unified and convenient manner, and
provides insight into the system designs.
%dependence of the performance on the system configuration.
%However, to the best of our knowledge, no work has been done to analyze the
% FER performance for the hybrid forward scheme with GRS in wireless relay networks.
%However, to the best of our knowledge, all these literature are limited in the performance of average symbol error rate (SER) analysis and no analytical average FER expressions in closed-form for the hybrid forward scheme with GRS   have been reported.
However, to the best of our knowledge, no results for  the HF scheme with the SC have been reported in the literature so far. Furthermore, most earlier works have been focusing on SER performance analysis   and limited works for the analytical frame error rate~(FER) performance of the HF scheme, which is more frequently used in the evaluation of the system performance, have been presented.
These motivated our work.

In this paper, we consider a two-hop network with one source node, one destination node and $N$ relay nodes, and with links experiencing independent block Rayleigh fading.
The main contributions of the present work can be briefly summarized
as follows:

% enumerate
% 3 items total
%% ----------------------
\begin{enumerate}[\IEEEsetlabelwidth{3)}]
\item
%Unlike~\cite{Duong2009P1,Ikki2009P1,Ikki2010P1},
We propose and study the
selective combining in hybrid cooperative networks~(SCHCN scheme) on a per frame basis. In the proposed SCHCNs scheme, each relay adaptively chooses, frame by frame, between the AF and DF protocols by examining the error-detecting code results. Specifically, we use cyclic redundancy check~(CRC) codes for the AF and DF adaptation in this work. The relays, which decode correctly,  will utilize  the DF protocol, and the rest will operate in the AF mode. Meanwhile, $N_c$ relays that have the highest effective  SNRs at the destination are chosen out of $N$ relays, and then retransmit their received signals to the destination.

\item
%To analyze the FER performance,
We develop an SNR threshold-based FER approximation model  %\footnote{In \cite{Amarasuriya2010P3091}, the authors analyzed the outage probability, but not FER.},
and then establish the relationship between the outage probability and the FER  %for general diversity systems
by using a novel analytical SNR threshold.

\item
Using the proposed FER approximation model, a closed-form FER expression  as well as a high SNR asymptotic FER expression for the SCHCN scheme are derived. The derived FER expressions clearly indicate that the SCHCN scheme can achieve full diversity order. Analytical results are verified through simulations. Simulation results show that the best RS scheme outperforms the multiple RS scheme under a total transmit power constraint\footnote{In some scenarios, such as sensor networks, energy efficiency emerges as a critical issue, and we can use the total transmit power constraint to portray the actual limitation.}, whereas the multiple RS scheme provides better performance than the best RS  scheme under an individual transmit power constraint\footnote{Here, we also investigate the case with individual transmit power constraint because in some relay scenarios, for example, fixed relay scenario\cite{Soldani2008P58}, the battery problem is not so sensitive as the mobile user relay scenario.}.
\end{enumerate}
%% ----------------------

%% Organization
The rest of the paper is organized as follows: we describe the system model and the SCHCN scheme in Section \ref{sec:system_model}. In Section \ref{sec:FER_approx_model}, we  develop an SNR threshold-based FER approximation model  and derive an improved SNR threshold. In Section \ref{sec:FER_GRS}, we calculate the closed-form average FER and the simplified asymptotic FER expressions    for the SCHCN scheme by using the FER approximation model. In Section \ref{sec:simulation}, we present some  simulation results. In Section \ref{sec:conclusion}, we draw the main conclusions.

%% Notation
\emph{\textbf{Notation}}:
For a random variable~(RV) $X$, $Pr(\cdot)$ represents its probability, $f_X(\cdot)$ represents its probability density function (PDF) and $F_X(\cdot)$ denotes its cumulative density function (CDF). Let $\mathcal{L}_{X}(s)$ and $\hat{\mathcal{L}}_{X}(s)$ stand for the Laplace transform of the PDF and CDF of $X$, respectively.
$X\sim \mathcal{CN}(0, \Omega)$ denotes a circular symmetric complex Gaussian variable with a zero mean and variance $\Omega$.  $Q(x)$ represents the $Q$-function, that is, $Q(x)=\frac{1}{\sqrt{2\pi}}\int_x^{\infty}e^{-t^2/2}\,\mathrm{d}t$.

%% Introduction section
%% ==================================== %%

%% ==================================== %%
%% Section TEMPLATE

%\section{Section Heading Here}\label{sec:xx}
%Section text here.
%
%\subsection{Subsection Heading Here}\label{subsec:xx.yy}
%Subsection text here.
%
%\subsubsection{Subsubsection Heading Here}\label{subsubsec:xx.yy.zz}
%Subsubsection text here.
%
%\paragraph{Paragraph Heading Here}\label{para:xx.yy.zz.aa}
%Paragraph text here.

%% Section TEMPLATE
%% ==================================== %%

%% ==================================== %%
%% System Model section

\section{System Model and the SCHCN Scheme}\label{sec:system_model}
%---------------------------------------
\subsection{System Model}  %(The Broadcasting Phase )
We consider a general two-hop wireless relay network with $N+2$ terminals consisting of one source node, $N$ relay nodes and one destination node, as shown
in Fig. \ref{fig:SystemDiagram_GRS}.
The channels of all links are assumed to be  quasi-static Rayleigh fading, i.e., the channel is fixed within one frame and changes independently from one frame to another. Let $h_0$, $h_{1,i}$  and $h_{2,i}$ represent the channel coefficients of the source-destination, source-relay~$i$ and relay~$i$-destination channels, respectively. Let $h_0\sim \mathcal{CN}(0,\Omega_0)$, $h_{1i} \sim \mathcal{CN}(0,\Omega_{1i})$ and $h_{2i} \sim
\mathcal{CN}(0,\Omega_{2i})$. Let $n_0$, $n_{1i}$ and
$n_{2i}$ represent the corresponding additive white Gaussian noises (AWGNs). It is  assumed that $n_0 \sim \mathcal{CN}(0,N_0)$, $n_{1i} \sim \mathcal{CN}(0,N_0)$  and $n_{2i} \sim \mathcal{CN}(0,N_0)$. Finally, let $\mathcal{E}_{s}$ and $\mathcal{E}_{r,i}$ denote the average transmit power at the source and the $i$th relay, respectively.

%In addition, we assume all channel state information (CSI) needed for decoding is
%available at the relay nodes and the destination node.  For simplicity, each terminal in the network is equipped
%with one single antenna working in the half-duplex mode, and signals are transmitted over orthogonal channels, either through time or frequency division.

%---------------------------------------
\subsection{Hybrid Forward Scheme}
%For simplicity,
It is assumed that each terminal in the network is in the half-duplex mode with single antenna.
Signals are transmitted over orthogonal channels, either through time or frequency division, in order to prevent interference.
Thus, one frame transmission in the relaying scheme consists of two separate phases, i.e., the broadcasting~($1$st) phase and the relaying~($2$nd) phase.

%---------------------------------------
\subsubsection{The Broadcasting Phase} %{The First Phase} %(The Broadcasting Phase )
The source broadcasts a signal to both the $N$ relays and the destination in the broadcasting  phase.
The received signals at the destination and the $i$th relay, at time $t$, denoted by $y_0(t)$ and $y_{1i}(t)$, can be expressed as
\begin{equation}
y_0(t) = \sqrt{\mathcal{E}_s}h_0 s(t) + n_0(t),
\end{equation}
and
\begin{equation}
y_{1i}(t) = \sqrt{\mathcal{E}_s} h_{1i} s(t) +n_{1i}(t),
\end{equation}
respectively, where $s(t)$ is the signal transmitted   at the source. The corresponding instantaneous SNRs
are given as
$\gamma_0 =\frac{\mathcal{E}_s}{N_0} |h_0|^2 = \tilde{\gamma}_0 |h_0|^2$  and $\gamma_{1i} = \frac{ \mathcal{E}_s }{N_0}|h_{1i}|^2 = \tilde{\gamma}_0 |h_{1i}|^2$,  respectively, where $\tilde{\gamma}_0 = \frac{\mathcal{E}_s}{N_0}$.
%\begin{equation} \label{eq:SNR_S2D}
%\gamma_0 =\frac{\mathcal{E}_s}{N_0} |h_0|^2 = \tilde{\gamma}_0 |h_0|^2 ,
%\end{equation}
%and
%\begin{equation} \label{eq:SNR_S2Ri}
%\gamma_{1i} = \frac{ \mathcal{E}_s }{N_0}|h_{1i}|^2 = \tilde{\gamma}_0 |h_{1i}|^2,
%\end{equation}
%respectively, where $\tilde{\gamma}_0 $ is defined as
%\begin{equation}
%\tilde{\gamma}_0 = \frac{\mathcal{E}_s}{N_0}.
%\end{equation}

%---------------------------------------
\subsubsection{The Relaying Phase}  %{The Second Phase} % (The Relaying Phase)
During the relaying phase, each relay adaptively chooses, frame by frame, between the AF and DF protocols by examining the error-detecting code result.  We use CRC codes for the AF and DF adaptation in this paper. %\cite{Song2009P1}
The relays which decode correctly, i.e., the CRC checking result is correct, are included in the DF group, denoted by $\mathcal{G}_{DF}$, and the rest are included in the AF group, denoted by $\mathcal{G}_{AF}$.

The corresponding received signal at the destination, denoted by $y_{2i}(t)$ can be written as
\begin{equation}
y_{2i}(t) = \sqrt{\mathcal{E}_{r,i}} h_{2i} x_i(t) +n_{2i}(t),
\end{equation}
where $x_i(t)$ represents the signal transmitted from the $i$th relay. Note that the transmit  signal $x_i(t)$
satisfies the following power constraint, $\mathbb{E}(\lvert x_i(t) \rvert ^2)\leq 1$.
%\begin{equation} \label{eq:x_i_power_constraint}
%\mathbb{E}(\lvert x_i(t) \rvert ^2)\leq 1.
%\end{equation}
Hence, the corresponding transmit  signal $x_i(t)$  can be expressed as\cite{Laneman2004P3062}
\begin{equation}
\begin{split}
x_i(t) = \begin{cases}
\frac{y_{1i}(t)}{\sqrt{\mathcal{E}_s |h_{1i}|^2+N_0}} & \text{if $i \in \mathcal{G}_{AF},$ } \\
s(t) & \text{if $ i \in \mathcal{G}_{DF}.$ }
\end{cases}
\end{split}
\end{equation}

Therefore the instantaneous SNR of the link through the $i$th relay node, denoted
by $\gamma_i$, is given as \cite{Liu2012P779}%\cite{Liu2009P1,Liu2012P779}
\begin{equation} \label{eq:SNR_i_cases}
\begin{split}
\gamma_i &= \begin{cases}
\frac{\gamma_{1i}\gamma_{2i}}{\gamma_{1i}+\gamma_{2i}+1} & \text{if $i \in \mathcal{G}_{AF}$ }, \\
\gamma_{2i} & \text{if $ i \in \mathcal{G}_{DF}$ },
\end{cases} \\
%&= \left(\frac{\gamma_{1i}\gamma_{2i}}{\gamma_{1i}+\gamma_{2i}+1} \right)^{z_i}
%\left(\gamma_{2i} \right)^{1-z_i},
\end{split}
\end{equation}
where $\gamma_{2i}$ is the instantaneous SNR of the link between the $i$th relay and
the destination, which is given by $\gamma_{2i} = \frac{ \mathcal{E}_{r,i} }{N_0} |h_{2i}|^2 = \tilde{\gamma}_{2i}|h_{2i}|^2$  and  $\tilde{\gamma}_{2i} = \frac{\mathcal{E}_{r,i}}{N_0}$.
%\begin{equation}
%\gamma_{2i} = \frac{ \mathcal{E}_{r,i} }{N_0} |h_{2i}|^2 = \tilde{\gamma}_{2i}|h_{2i}|^2,
%\end{equation}
%and $\tilde{\gamma}_{2i} $ is defined as
%\begin{equation}
%\tilde{\gamma}_{2i} = \frac{\mathcal{E}_{r,i}}{N_0}.
%\end{equation}

%---------------------------------------
\subsection{SCHCN Scheme}

On receiving the signal sent from the source, each relay needs to send one bit indicator to inform the destination that it uses the AF or DF protocols. In the SCHCN scheme, destination then calculates the overall received SNR from each relay accordingly and selects the $N_c$ relays with the largest effective SNRs. After finding the optimum $N_c$ relays, through a feedback channel, the destination will inform which relays are selected for transmission during the relaying phrase and other unselected relays will be in an idle state.
Let $\gamma_{(1)}\geq \gamma_{(2)}\geq\dotsb\geq \gamma_{(N)} $ represent the decreasing ordered SNRs of the $\{\gamma_n\}_{n=1}^{N}$. Then, the instantaneous SNR of the  MRC  output at the destination is given by\footnote{It it noteworthy that, during the pilot phase, the destination has to estimate $N$ channels from the relay nodes and one channel from the source, and then selects the $N_c$ best relay nodes. Thus the complexity of the pilot phase, is mainly dominated by the $N+1$ channel estimation process. The exact complexity is dependent on the channel estimation approach, for example, when using zero-forcing~(ZF) to estimate the channel, one needs $N+1$ floating-point operations. Furthermore, during the data transmission phase, the destination has to combine $N_c$ links from the relay nodes as well as the direct link from the source by using the MRC technique. Thus, the complexity of the data combining process requires $N_c+1$ multiplies and $N_c$ additions, i.e., $2N_c+1$ floating-point operations.}\cite{Rappaport2002P}
% In the SCHCN scheme,  $N_c$ relays, which have the $N_c$ highest effective SNRs at the destination, will forward the received signal to the destination during the relaying phase and other unselected relays will be in an idle state.

\begin{equation}\label{eq:SNR_GRS_MRC}
\begin{split}
\gamma_{SCHCN} = \gamma_{0} + \sum_{n=1}^{N_c}\gamma_{(n)}=\gamma_{0}+\gamma_{coop},
\end{split}
\end{equation}
where $\gamma_{coop}=\sum_{n=1}^{N_c}\gamma_{(n)}$.
%\begin{equation}\begin{split}\label{eq:SNR_coop}
%\gamma_{coop}=\sum_{n=1}^{N_c}\gamma_{(n)}.
%\end{split}\end{equation}

It is  assumed that  all channel state information (CSI) needed for decoding is
available at the relay nodes and the destination node\cite{Laneman2004P3062,Wang2007P1427,Vien2009P2849}. This can be easily achieved in practice. For example, the relay acquires the CSI of the source-relay channel via the pilot symbol sent from the source and then transmits it to the destination via a feedback channel. Similarly, the destination can acquire the CSI of the relay-destination via the pilot symbol sent from the relay.
%The selection process in the SCHCN scheme is exactly the
%same as in the conventional AF or DF selection schemes and
%does not add any extra complexity in system implementation.
%The only requirement for the SCHCN scheme is that each relay
%needs send one bit indicator to inform the destination that it
%use the AF or DF protocol. Destination then calculates the
%overall received SNR from each relay accordingly and selects
%the $N_c$ relays with the largest SNR.
For simplicity, let $\lambda_{0} = \frac{1}{\tilde{\gamma}_0 \Omega_{0}}$, $\lambda_{1i} = \frac{1}{\tilde{\gamma}_0 \Omega_{1i}}$, and $ \lambda_{2i} = \frac{1}{\tilde{\gamma}_{2i} \Omega_{2i}}$.
As $\gamma_0$, $\gamma_{1i}$ and $\gamma_{2i}$ are exponentially
distributed, their PDFs  can be expressed as
\begin{equation}
f_{\gamma_k}(\gamma) = \lambda_{k} e^{- \lambda_{k} \gamma }, k \in
\{0,1i,2i\}.
\end{equation}
It is  also assumed that $\{\gamma_{1i}\}_{i=1}^{N}$  ($\{\gamma_{2i}\}_{i=1}^{N}$) are independent and identically distributed (i.i.d.) RVs and let $\lambda_{1i}= \lambda_{sr}$ and $\lambda_{2i} =\lambda_{rd}$ for analysis tractability in this work.
%and
%\begin{equation}
%F_{\gamma_k}(\gamma) = 1 - e^{- \lambda_{k} \gamma }, k \in
%\{0,1i,2i\},
%\end{equation}
%respectively.

%% System Model section
%% ==================================== %%

%% ==================================== %%
%% FER Model section

\section{SNR Threshold-Based FER Approximation Model} \label{sec:FER_approx_model}
In this section,  an SNR threshold-based FER approximation model is proposed for %general
diversity systems. As we know, cooperative systems can be viewed as a special kind of general diversity systems \cite{Laneman2004P3062}. Therefore, we can apply the results obtained in this section to the FER analysis of the SCHCN scheme presented in next section.

%Note that the core idea of the SNR threshold-based FER approximation is to model the FER using a outage probability model with an appropriate SNR threshold.
In the following, the SNR threshold-based FER approximation model is described in Subsection \ref{subsec:FER_model_introduction},  some existing related works are presented in Subsection \ref{subsec.previous.SNR}, and then   an improved criterion and the corresponding SNR threshold are presented in Subsection \ref{subsec:improved_SNR}.
Finally,  a multiple-input multiple-output~(MIMO) system using the proposed SNR threshold-based FER approximation model is presented as a simple example in Subsection \ref{subsec:mimo}.

\subsection{SNR threshold-based FER approximation model} \label{subsec:FER_model_introduction}

Note that the average FER, represented by $\bar{P}_f$,
can be calculated by integrating the instantaneous FER in AWGN
channel, denoted by $P_f^G(\gamma)$, over the fading
distribution \cite{Proakis2001P}
\begin{equation} \label{eq:P_f_B}
\bar{P}_f(\bar{\gamma}) = \int_{0}^{\infty} P_f^G(\gamma) f_{\gamma}(\gamma,\bar{\gamma})
\text{d} \gamma,
\end{equation}
where $\gamma$ and $\bar{\gamma}$ represent  the instantaneous and average SNR, respectively, and $f_{\gamma}(\cdot)$ represents
the PDF of $\gamma$.

%Although
Eq. (\ref{eq:P_f_B}) is an exact expression to calculate $\bar{P}_f $; however,
its closed-form expression is still difficult to be derived. The SNR threshold-based FER approximation model is an  accurate and simple approach to evaluate FER. Assume that the instantaneous FER is $1$ when the
instantaneous SNR $\gamma$ is below a given  SNR threshold $\gamma_t$,
otherwise it is $0$ \cite{ElGamal2001P671, Chatzigeorgiou2008P577,Chatzigeorgiou2009P216}, i.e.,
\begin{equation}
P_f^G(\gamma|\gamma \leq \gamma_t) \approx 1 \; \text{and} \;
P_f^G(\gamma|\gamma > \gamma_t) \approx 0 .
\end{equation}
Then, the average FER is given  as\cite{ElGamal2001P671,Chatzigeorgiou2008P577,Chatzigeorgiou2009P216}
\begin{equation} \begin{split} \label{eq:ThresholdModelFER}
\bar{P}_f(\bar{\gamma})
 &= \int_{0}^{\gamma_t} P_f^G(\gamma) f_{\gamma}(\gamma,\bar{\gamma})\text{d} \gamma
+ \int_{\gamma_t}^{\infty} P_f^G(\gamma) f_{\gamma}(\gamma,\bar{\gamma})\text{d} \gamma \\
&\approx \int_{0}^{\gamma_t} f_{\gamma}(\gamma,\bar{\gamma}) \text{d} \gamma =
F_{\gamma}(\gamma_t,\bar{\gamma}),
\end{split} \end{equation}
where  $F_{\gamma}(\cdot)$ represents the
CDF of $\gamma$.

%Note that,
According to the SNR threshold-based FER approximation model, the average
FER is approximately calculated as an outage probability. Hence the accuracy of the FER model is principally determined by the value of the given SNR threshold.% selection.

\subsection{Related Works}\label{subsec.previous.SNR}

In \cite{ElGamal2001P671}, %El Gamal and Hammons
the SNR threshold-based FER approximation model is applied to the iteratively decoded systems with turbo codes. It is shown that the optimal SNR threshold coincides with the convergence threshold of the iterative turbo decoder. %However, the result can be only used for system with applying turbo channel codes.
Furthermore, the SNR threshold-based FER model is extended to non-iterative coded and uncoded systems in \cite{Chatzigeorgiou2008P577,Chatzigeorgiou2009P216}. %To get the optimal SNR threshold, a proper error criterion should be used.
The \emph{minimum absolute error sum criterion} is adopted in \cite{Chatzigeorgiou2008P577,Chatzigeorgiou2009P216} to minimize the sum of absolute error
\begin{equation} \label{eq:AbsoluteCriterion}
\gamma_{t} = \text{arg} \; \min_{\gamma}\left\{ \int_{0}^{\infty} \left|
\bar{P}_{f}(\bar{\gamma}) - F_{\gamma}(\gamma,\bar{\gamma})\right|
\text{d} \bar{\gamma} \right\} \;,
\end{equation}
where $\bar{\gamma}$ denotes the average SNR, and the SNR threshold is given as
\begin{equation} \label{eq.SNR.threshod.old}
\gamma_t = \left(\int_{0}^{\infty} \frac{1-P_f^G(\gamma)}{\gamma^2}
\text{d} \gamma \right)^{-1}.
\end{equation}

Since the \emph{minimum absolute error sum criterion} does not consider the fact that the FER decreases more quickly at high SNR region in high diversity order systems, it needs to be improved for general diversity order systems.
%We note that the model on the basis of the \emph{minimum absolute error sum criterion} might not be sufficiently accurate since it does not consider the fact that FER decreases more quickly at high SNR region in high diversity order systems. Hence, it can be improved.

\subsection{Improved SNR Threshold}\label{subsec:improved_SNR}
%We develop
A new SNR threshold-based FER approximation model is developed for the  diversity systems based on two important aspects.

Firstly, a \emph{minimum relative error
sum criterion}, which taking into account the fact that the FER
decreases quickly when SNR increases and minimizing the sum of \emph{relative} error, is used instead of the \emph{minimum absolute error sum criterion} in Eq. (\ref{eq:AbsoluteCriterion}) as
\begin{equation} \label{eq:RelativeCriterion}
\gamma_{t} = \text{arg} \; \min_{\gamma}\left\{ \int_{0}^{\infty} \left|
\frac{ \bar{P}_{f}(\bar{\gamma}) -
F_{\gamma}(\gamma,\bar{\gamma})}{\bar{P}_{f}(\bar{\gamma})}\right|
\text{d} \bar{\gamma} \right\} \; .
\end{equation}

In the Appendix \ref{Proof:SNR_threshold}, we will show that the \emph{minimum relative error
sum criterion} will lead to a \emph{minimum relative
error criterion at high SNR} as
\begin{equation}
\label{eq:ZECriterion} \gamma_{t} = \text{arg} \lim_{\substack{\gamma},\bar{\gamma}
\rightarrow \infty   } \left\{\bar{P}_{f}(\bar{\gamma})
-F_{\gamma}(\gamma,\bar{\gamma}) =0 \right\}.
\end{equation}

Secondly, the calculation of SNR threshold should take into account the factor of the
diversity order. We note that  in a Rayleigh fading channel, the CDF of SNR is given by $F_{\gamma}(\gamma, \bar{\gamma}) = 1- e^{-\gamma / \bar{\gamma}}$, then we have $\lim_{\bar{\gamma} \rightarrow \infty }F_{\gamma}(\gamma, \bar{\gamma})
= {\bar{\gamma}}^{-1} \gamma $.
 Similarly,  in a Nakagami-$m$ fading channel, the CDF of SNR is   $F_{\gamma}(\gamma, \bar{\gamma}) = \frac{\Gamma(m,m\gamma/\bar{\gamma})}{\Gamma(m)}$, where $\Gamma(m)$ denotes the Gamma function, and $\Gamma(m,x)$ denotes the lower part incomplete Gamma function given by  $\Gamma(m,x)=\int_{0}^x t^{m-1}e^{-t}\text{d}t$.
 At high SNR, we have    $\lim_{\bar{\gamma} \rightarrow \infty }F_{\gamma}(\gamma, \bar{\gamma})
= \Gamma(m)^{-1}\int_{0}^{m\gamma/\bar{\gamma}} t^{m-1}\text{d}t =\frac{m^{m-1}}{ \Gamma(m) (\bar{\gamma})^m } \gamma^m $. Generally speaking,   the CDF of a system with a
diversity order of $d$ at high SNR can be approximated as \cite{Chatzigeorgiou2009P216,Zheng2003P1073,Rodrigues2008P449}:
\begin{equation} \label{eq:CDFLimit}
\lim_{\bar{\gamma} \rightarrow \infty}F_{\gamma}(\gamma, \bar{\gamma})
\rightarrow G(\bar{\gamma}) \gamma ^d,
\end{equation}
where $G(\bar{\gamma})$ represents a constant related to $\bar{\gamma}$.

Combining Eq. (\ref{eq:P_f_B}), Eq. (\ref{eq:ZECriterion})
and Eq. (\ref{eq:CDFLimit}), we derive the closed-form improved SNR threshold as
\begin{equation} \label{eq:SNR_threshold_continue}
\begin{split}
\gamma_{t,d} &= \text{arg} \lim_{\gamma_t,\bar{\gamma} \rightarrow \infty}
\left\{ \int_{0}^{\infty} f_{\gamma}(\gamma, \bar{\gamma})
P_f^G(\gamma) \text{d} \gamma
- G(\bar{\gamma}) \gamma_t ^d= 0\right\} \\
&=\text{arg} \lim_{\gamma_t,\bar{\gamma} \rightarrow \infty} \left\{
\int_{0}^{\infty} G(\bar{\gamma}) d \gamma ^{d-1} P_f^G(\gamma)
\text{d} \gamma - G(\bar{\gamma}) \gamma_t ^d
= 0\right\} \\
&= \text{arg} \lim_{\gamma_t,\bar{\gamma} \rightarrow \infty} \left\{ {\gamma_t}^d = d \int_{0}^{\infty} \gamma^{d-1} P_f^G(\gamma) \text{d} \gamma \right\} \\
&= \left( d \int_{0}^{\infty} \gamma^{d-1} P_f^G(\gamma) \text{d}
\gamma \right)^{1/d}.
\end{split}\end{equation}
%Interestingly, when the diversity order $d$ is $1$, using Eq. (\ref{eq:SNR_threshold_continue}), the improved SNR threshold is given by
%\begin{equation} \label{eq:SNR_threshold_1}
%\gamma_{t,1} = \int_{0}^{\infty} P_f^G(\gamma) \text{d} \gamma.
%\end{equation}

Note that $ P_f^G(\gamma) $ can be calculated in closed-form for uncoded systems. For example, for linear modulation, $ P_f^G(\gamma) $ is given as\cite{Proakis2001P,Chatzigeorgiou2008P577}
\begin{equation} \label{eq:P_f_G_uncode}
P_f^G(\gamma) = 1 - \left(1 - Q\left(\sqrt{c \gamma}\right)\right)^L,
\end{equation}
where $c$ represents modulation constant ($c=2$ for binary-phase-shift-keying (BPSK)), $Q(\cdot)$ represents the $Q$ function, and $L$ denotes the frame length.

Substituting Eq. (\ref{eq:P_f_G_uncode}) into Eq. (\ref{eq:SNR_threshold_continue}),   the SNR threshold $\gamma_{t,d}$ for uncoded systems is given as:
\begin{equation} \label{eq:SNR_threshold_uncoded}
\gamma_{t,d} = \left( d \int_{0}^{\infty} \gamma^{d-1} \left(1 - \left(1 - Q\left(\sqrt{c \gamma}\right)\right)^L \right) \text{d}
\gamma \right)^{1/d}.
\end{equation}
In Table~\ref{table:SNR threshold}, we present the SNR thresholds by using Eq. (\ref{eq:SNR_threshold_uncoded}) for the systems with different diversity orders when  $c=2$. %, $L=100$.

%For coded systems, we can calculate Eq. (\ref{eq:SNR_threshold_continue}) using numerical method. We can first get $P_f^G(\gamma) $, i.e. the instantaneous FER for the scheme over AWGN channel, using Monte Carlo methods. We assume that the SNR values $\gamma_i$,
%$i=1,2,\cdots,N$, are equally spaced with $\Delta \gamma$ and
%ordered. Then we can obtain the following equivalent expression for
%discrete SNR values
%\begin{equation}\label{eq:SNR_threshold_quanti}
%\gamma_{t,d} = \left( d \sum_{i=1}^{N}\gamma_i^{d-1} P_f^G(\gamma_i)
%\Delta \gamma\right)^{1/d} .
%\end{equation}

We note that when substituting $c=2$ (BPSK), $L=1$ and $d=1$  into Eq. (\ref{eq:SNR_threshold_uncoded}) and with the help of \cite[Eq. (6.313.2)]{Gradshteyn2007P}, we have $\gamma_{t,1}=\frac{1}{4}$  and $\bar{P}_f = 1-e^{-\frac{1}{4\bar{\gamma}}} \approx \frac{1}{4\bar{\gamma}}$, which is equivalent to the  average BER of BPSK over Rayleigh fading channel at high SNR: $\bar{P}_b=\frac{1}{2}(1-\sqrt{\frac{\bar{\gamma}}{1+\bar{\gamma}}}) \approx \frac{1}{4\bar{\gamma}}$ \cite{Proakis2001P}. %It indicates that our proposed model is accurate even for very short frames.

\subsection{Simple Example of MIMO Systems}\label{subsec:mimo}

Let us use a  MIMO system as a simple example to illustrate the  usage of the developed SNR threshold-based FER model.
Assume that the MIMO system has $N_T$ inputs and $N_R$ outputs, and the transmitter uses space-time block coding~(STBC) \cite{Alamouti1998P1451}. The receiver use MRC technique to combine the $N=N_T N_R$ independent fading paths, and the PDF and CDF of the instantaneous SNR, represented by $\gamma$, at the output of the combiner are given by\cite{Proakis2001P,Chatzigeorgiou2009P216}:
\begin{equation}
f_{\gamma}(\gamma, \bar{\gamma}) = \frac{\gamma^{N-1}e^{-\gamma/(\bar{\gamma}/N_T)}}{(\bar{\gamma}/N_T)^N (N-1)!},
\end{equation}
and
\begin{equation}\label{eq:MIMO.CDF}
\begin{split}
F_{\gamma}(\gamma, \bar{\gamma})
&= 1- e^{-\gamma N_T/\bar{\gamma}}\sum_{k=0}^{N-1}\frac{(\gamma N_T /\bar{\gamma})^k}{k!},
\end{split}
\end{equation}
respectively, where $\bar{\gamma}$ represents the average SNR per receive antenna.

It should be mentioned here that the approximation in Eq. (\ref{eq:CDFLimit}) still holds for the MIMO system.
\begin{IEEEproof}%[Proof of Theorem \ref{thm:SVD}]
With the fact that $\lim_{\bar{\gamma} \rightarrow \infty}e^{-\gamma N_T/\bar{\gamma}}=1$, we have
\begin{equation}
 \begin{split}
\lim_{\bar{\gamma} \rightarrow \infty} F_{\gamma}(\gamma, \bar{\gamma})
& = \lim_{\bar{\gamma} \rightarrow \infty}\int_0^{\gamma} f_{\gamma}(t, \bar{\gamma}) \text{d} t\\
& = \int_0^{\gamma} \frac{t^{N-1}}{(\bar{\gamma}/N_T)^N (N-1)!} \text{d} t\\
&=\frac{1}{(\bar{\gamma}/N_T)^N N!} \gamma^{N}\\
&
=G(\bar{\gamma}) \gamma^{N},
\end{split}
\end{equation}
where $G(\bar{\gamma}) =\frac{1}{(\bar{\gamma}/N_T)^N N!}$.
\end{IEEEproof}
In the following, with the  derived  SNR threshold $\gamma_{t,N}$   in Subsection \ref{subsec:improved_SNR},  the
FER of the MIMO system in quasi-static fading channel  is
approximately calculated as
\begin{equation}
\begin{split}
\bar{P}_f(\bar{\gamma}) &\approx   F_{\gamma}(\gamma_{t,N}, \bar{\gamma})\\
&= 1- e^{-\gamma_{t,N}N_T/\bar{\gamma}}\sum_{k=0}^{N-1}\frac{(\gamma_{t,N} N_T /\bar{\gamma})^k}{k!}.% \\
% &< \frac{1}{(\bar{\gamma}/N_T)^N N!} \gamma_{t,N}^{N}.
\end{split}
\end{equation}

In the simulation section, we compare the proposed SNR threshold obtained in Eq. (\ref{eq:SNR_threshold_continue}) and the optimal SNR threshold obtained in Eq. (\ref{eq:opt.SNR.threshold})\footnote{ $F^{-1}_{\gamma}(\, \cdot \, , \bar{\gamma})$  represents the inverse function of  $F_{\gamma}(\gamma, \bar{\gamma})$ with respect to $\gamma$.} by using numerical method for the MIMO system. Simulation results show that the proposed SNR threshold converges well with the optimal SNR threshold especially at medium and high SNRs, which justifies the proposed FER approximation model.
\begin{equation}\label{eq:opt.SNR.threshold}
\begin{split}
 \gamma^{opt}_{t,N} &=   F^{-1}_{\gamma}(\bar{P}_f(\bar{\gamma}), \bar{\gamma}).\\
\end{split}
\end{equation}

So far,  an improved SNR threshold-based FER approximation model is developed
for   diversity systems. In this developed FER model, the average FER is approximately calculated as an
outage probability as Eq. (\ref{eq:ThresholdModelFER}), the improved SNR threshold can be calculated using  Eq. (\ref{eq:SNR_threshold_continue}) for general systems and can be calculated using Eq. (\ref{eq:SNR_threshold_uncoded})  for uncoded systems. % and  using Eq. (\ref{eq:SNR_threshold_quanti}) for coded systems.
In the next section, we will apply this model to the FER analysis of the SCHCN scheme for  cooperative systems.
\emph{We should emphasize} that our analysis in this work can be easily extended to the coded systems. However, the corresponding details for coded systems are omitted here because of limited space and please refer to \cite{Liu2012P779,Huo2010P1} for details.

%% FER Model section
%% ==================================== %%

%% ==================================== %%
%% Analysis section

\section{FER Analysis}\label{sec:FER_GRS}
%\section{FER Analysis of the SCHCN scheme}\label{sec:FER_GRS}
%In this section,  we perform the FER analysis of the SCHCN scheme in cooperative systems.
We first apply the SNR threshold-based FER approximation model to calculate  the closed-form analytical FER of the SCHCN scheme in Subsection \ref{subsec:FER_HF_GRS}.
We then present the asymptotic FER expression for the SCHCN scheme  in Subsection \ref{subsec:FER_HF_GRS_at_High_SNR}.
%Furthermore, we extend the analytical results to other traditional relaying schemes in Subsection \ref{subsec:Extend}.

%%****************************************************************
\subsection{Closed-form analytical FER  }\label{subsec:FER_HF_GRS} %of the SCHCN Scheme
The  following properties for the SCHCN scheme are first established before performing FER calculation.
%%-------------------------theorem-------------------------
\begin{theorem} \label{theorem:CDF_HF_SNR}
The CDF of the end-to-end SNR of the indirect link $S\rightarrow R_i \rightarrow D$ for the SCHCN scheme can be expressed as  %is given by
\begin{equation} \begin{split} \label{eq:CDF_HF_r_i}
F_{\gamma_i}(\gamma ) &\approx
\begin{cases}
1- e^{-(\lambda_{sr}+\lambda_{rd})\gamma }, &\text{if } \gamma < \gamma_{t,1},\\
1-e^{-\lambda_{sr} \gamma_{t,1}}e^{-\lambda_{rd}\gamma}, &\text{if } \gamma \geq \gamma_{t,1}.
\end{cases}
\end{split} \end{equation}
For simplicity, Eq. (\ref{eq:CDF_HF_r_i}) can be further approximated by  an exponential function to extend the range of $\gamma$ from $\gamma_{t,1}$ to $\gamma_{t,d}$
\begin{equation}  \label{eq:CDF_HF_r_i_Approximation}
F_{\gamma_i}(\gamma )  \approx
1- e^{- \lambda_{eq} \gamma }, \; 0\leq\gamma\leq\gamma_{t,d},
\end{equation}
where the lower bound and upper bound of $\lambda_{eq}$ are $\lambda_{sr} \gamma_{t,1}/\gamma_{t,d} +\lambda_{rd}$ and $\lambda_{sr} +\lambda_{rd}$, respectively. From Table~\ref{table:SNR threshold}, we note that $\gamma_{t,1}/\gamma_{t,d}\approx 1$ which implies that either the lower bound or upper bound can be used to approximate Eq. (\ref{eq:CDF_HF_r_i}) and the approximation  is tight.
\end{theorem}
\begin{proof}[Proof]
Please see Appendix \ref{Proof:CDF_HF_SNR}.
\end{proof}

%%-------------------------theorem-------------------------
\begin{theorem} \label{theorem:math_PDF}
 Let  $\{\gamma_{i}\}_{i=1}^{N}$ be  i.i.d.  exponential RVs, represented by $\gamma_i\sim\mathcal{E}(\lambda_{eq})$. Let $\gamma_{(1)}\geq \gamma_{(2)}\geq\dotsb\geq \gamma_{(N)}$ denote the order statistics obtained by arranging the $\{\gamma_n\}_{n=1}^{N}$ in a decreasing order.
According to Eq. (\ref{eq:SNR_GRS_MRC}), the CDF of $\gamma_{SCHCN}$ in closed-form expression is derived   for the two cases of both $\lambda_{0}\neq \lambda_{eq}$ and $\lambda_{0}=\lambda_{eq}$ as follows
%%----------paragraph------------------------
\paragraph{Case of $\lambda_{0}\neq \lambda_{eq}$}
\begin{equation} \begin{split}  \label{eq:CDF_SNR_GRS_case1}
F_{\gamma_{SCHCN} }(\gamma,\lambda_{eq}) &=\sum_{i=1}^{N_c} \frac{ \alpha_i } {\lambda_{eq} ^{i}}\left[ 1-e^{-\lambda_{eq}  \gamma}\sum_{m=0}^{i-1}  \frac{(\lambda_{eq} \gamma)^m}{m!}\right]  \\
&\quad +\sum_{j=1}^{N-N_c} \frac{\beta_j}{(1+\frac{j}{N_c})\lambda_{eq}}(1-e^{-(1+\frac{j}{N_c})\lambda_{eq}\gamma})  \\
&\qquad + \frac{\beta_0}{\lambda_0}(1-e^{-\lambda_{0}\gamma}), \\
\end{split} \end{equation}
%\begin{equation} \begin{split}  \label{eq:CDF_SNR_GRS_case1}
%F_{\gamma_{GSC} }(\gamma,\lambda_{eq}) &=1-\sum_{i=1}^{N_c} \frac{ \alpha_i } {\lambda_{eq} ^{i}} e^{-\lambda_{eq}  \gamma}\sum_{m=0}^{i-1}  \frac{(\lambda_{eq} \gamma)^m}{m!} - \frac{\beta_0}{\lambda_0} e^{-\lambda_{0}\gamma}  \\
%&\qquad\qquad-\sum_{j=1}^{N-N_c} \frac{\beta_j}{(1+\frac{j}{N_c})\lambda_{eq}} e^{-(1+\frac{j}{N_c})\lambda_{eq}\gamma} ,
%\end{split} \end{equation}
where\footnote{Note that as a special case of $N_c = N$, the following parameters are reduced to
$c= \lambda_{0}\lambda_{eq}^N$, $\theta_0 =1$, $\theta_k =0$,   $\alpha_i=\frac{c}{(N-i)!} \left[\frac{ \theta_0(-1)^{N-i}(N-i)!}{(-\lambda_{eq}+\lambda_{0})^{N-i+1}}  \right]$,   $\beta_0=\frac{c \cdot \theta_0 }{(\lambda_{eq}-\lambda_{0})^{N} }$, and $\beta_j =0$.}   $c=\frac{N!}{N_c!N_c^{N-N_c}}\lambda_{0}\lambda_{eq}^N$, $\theta_0=\frac{1}{{\prod_{n=1}^{N-N_c} [(1+\frac{n}{N_c})\lambda_{eq}-\lambda_{0}]}}$, $\theta_k=\frac{1}{[-(1+\frac{k}{N_c})\lambda_{eq}+\lambda_{0}]\prod_{n=1,n \neq k}^{N-N_c} [(\frac{n-k}{N_c})\lambda_{eq}]}$, $\alpha_i = \frac{c}{(N_c-i)!} \left[\frac{\theta_0(-1)^{N_c-i}(N_c-i)!}{(-\lambda_{eq}+\lambda_{0})^{N_c-i+1}}   +\sum_{k=1}^{N-N_c} \frac{\theta_k(-1)^{N_c-i}(N_c-i)!}{( \frac{k}{N_c} \lambda_{eq})^{N_c-i+1}}  \right]$, $\beta_0=\frac{c\cdot \theta_0}{(\lambda_{eq}-\lambda_{0})^{N_c} }$  and $\beta_j=\frac{c\cdot \theta_j}{  (-\frac{j}{N_c}\lambda_{eq})^{N_c} }$.
%\begin{equation} \begin{split}
%\alpha_i&= \frac{c}{(N_c-i)!} \left[\frac{\theta_0(-1)^{N_c-i}(N_c-i)!}{(-\lambda_{eq}+\lambda_{0})^{N_c-i+1}}   +\sum_{k=1}^{N-N_c} \frac{\theta_k(-1)^{N_c-i}(N_c-i)!}{( \frac{k}{N_c} \lambda_{eq})^{N_c-i+1}}  \right],
%\end{split} \end{equation}
%\begin{equation} \begin{split}
%\beta_0&==\frac{c\cdot \theta_0}{(\lambda_{eq}-\lambda_{0})^{N_c} },
%\end{split} \end{equation}
%\begin{equation} \begin{split}
%\beta_j&=\frac{c\cdot \theta_j}{  (-\frac{j}{N_c}\lambda_{eq})^{N_c} },
%\end{split} \end{equation}
%\begin{equation} \begin{split}
%c=\frac{N!}{N_c!N_c^{N-N_c}}\lambda_{0}\lambda_{eq}^N,
%\end{split} \end{equation}
%\begin{equation} \begin{split}
%\theta_0=\frac{1}{{\prod_{n=1}^{N-N_c} [(1+\frac{n}{N_c})\lambda_{eq}-\lambda_{0}]}} ,
%\end{split} \end{equation}
%and
%\begin{equation} \begin{split}
%\theta_k&=\frac{1}{[-(1+\frac{k}{N_c})\lambda_{eq}+\lambda_{0}]\prod_{n=1,n \neq k}^{N-N_c} [(\frac{n-k}{N_c})\lambda_{eq}]}.
%\end{split} \end{equation}

%%----------paragraph------------------------
\paragraph{Case of $\lambda_{0}=\lambda_{eq}$}
\begin{equation} \begin{split}\label{eq:CDF_SNR_GRS_case2}
F_{\gamma_{SCHCN} }(\gamma,\lambda_{eq}) &=\sum_{i=1}^{N_c+1} \frac{ \alpha_i^* } {\lambda_{eq} ^{i}}\left[ 1-e^{-\lambda_{0} \gamma}\sum_{m=0}^{i-1} \frac{(\lambda_{0} \gamma)^m}{m!}\right]  \\
&\quad +\sum_{j=1}^{N-N_c} \frac{\beta_j^*}{(1+\frac{j}{N_c})\lambda_{0}}(1-e^{-(1+\frac{j}{N_c})\lambda_{0}\gamma}), \\
\end{split} \end{equation}
where\footnote{It is noteworthy that  as a special case of $N_c = N$, the following parameters are reduced to
$c^*= \lambda_{0}^{N+1}$, $\theta_k^* =0$,   $\alpha_i^*=0$~($i\neq N+1$), $\alpha_{N+1}^*=c^*$,  and $\beta_j^* =0$.}  $c^*=\frac{N!}{N_c!N_c^{N-N_c}}\lambda_{0}^{N+1}$, $\theta_k^* =\frac{1}{ \prod_{n=1,n \neq k}^{N-N_c} [(\frac{n-k}{N_c})\lambda_{0}]}$, $\alpha_i^*  =\frac{c^*}{(N_c+1-i)!} \left[ \sum_{k=1}^{N-N_c} \frac{\theta_k^*(-1)^{N_c+1-i}(N_c+1-i)!}{( \frac{k}{N_c} \lambda_{0})^{N_c-i+2}} \right]$   and $\beta_j^* =\frac{c^*\cdot\theta_j^*}{ (-\frac{j}{N_c}\lambda_{0})^{N_c+1} }$.
%\begin{equation} \begin{split}
%c^*=\frac{N!}{N_c!N_c^{N-N_c}}\lambda_{0}^{N+1},
%\end{split} \end{equation}
%\begin{equation} \begin{split}
%\theta_k^*&=\frac{1}{ \prod_{n=1,n \neq k}^{N-N_c} [(\frac{n-k}{N_c})\lambda_{0}]},
%\end{split} \end{equation}
%\begin{equation} \begin{split}
%\alpha_i^*& =\frac{c^*}{(N_c+1-i)!} \left[ \sum_{k=1}^{N-N_c} \frac{\theta_k^*(-1)^{N_c+1-i}(N_c+1-i)!}{( \frac{k}{N_c} \lambda_{0})^{N_c-i+2}} \right],
%\end{split} \end{equation}
%and
%\begin{equation} \begin{split}
%\beta_j^*&=\frac{c^*\cdot\theta_j^*}{ (-\frac{j}{N_c}\lambda_{0})^{N_c+1} }.
%\end{split} \end{equation}
\end{theorem}
\begin{proof}[Proof]
Please see Appendix \ref{Proof:CDF_SNR}.
 \end{proof}

According to the SNR threshold-based FER approximation model in Section \ref{sec:FER_approx_model}, the average FER of the SCHCN scheme  can be expressed as an outage probability with an appropriate SNR threshold.  % Hence, to get the FER, we merely need to know the CDF of the SNR of the SCHCN scheme, denoted by $F_{\gamma_{SCHCN}}(\gamma,\lambda_{eq})$.
Thus,   the average FER of
the SCHCN scheme, denoted by $\bar{P}_f$, is given as
\begin{equation} \begin{split} \label{eq:P_f_HF_GRS}
\bar{P}_{f} & \approx   F_{\gamma_{SCHCN}}(\gamma_{t,d},\lambda_{eq}),
\end{split} \end{equation}
where $d$ denotes the diversity order of  the cooperative system using the SCHCN scheme and the calculation methods of parameters   $\gamma_{t,d}$ are presented in Section \ref{sec:FER_approx_model}.
The accuracy of  the derived close-form  FER expression, i.e., Eq.~(\ref{eq:P_f_HF_GRS}), is verified through simulation results in the next section.%
%is accurate especially at medium and high SNR values.
% which is validated by simulation results.

%%****************************************************************
\subsection{Asymptotic FER }\label{subsec:FER_HF_GRS_at_High_SNR}
In this subsection, we derive the asymptotic  FER for the SCHCN scheme at high SNR.
\begin{theorem} \label{theorem:math_PDF_simplified}
 The asymptotic CDF of $\gamma_{SCHCN}$ at high SNR   can be written in closed-form as
\begin{equation} \begin{split} \label{eq:CDF_high_SNR}
F_{\gamma_{SCHCN}}(\gamma,\lambda_{eq}) \approx \frac{\lambda_{0}{\lambda_{eq}}^N\gamma^{N+1}}{(N+1)N_c!{N_c}^{N-N_c}}  .
\end{split}  \end{equation}
\end{theorem}

\begin{proof}[Proof]
Please see Appendix \ref{Proof:CDF_SNR_at_high_SNR}.
\end{proof}
With the help of Property \ref{theorem:math_PDF_simplified}, we can obtain the asymptotic FER expression of the SCHCN scheme as
\begin{equation} \begin{split} \label{eq:P_f_HF_high_SNR}
\bar{P}_{f} & \approx   F_{\gamma_{SCHCN}}(\gamma_{t,d},\lambda_{eq}) \approx  \frac{\lambda_{0}{\lambda_{eq}}^N\gamma_{t,d}^{N+1}}{(N+1)N_c!{N_c}^{N-N_c}}.\\
\end{split}  \end{equation}

From Eq. (\ref{eq:P_f_HF_high_SNR}) it shows that  full diversity order of $N+1$ can be achieved in the SCHCN scheme, and thus, $\gamma_{t,d}=\gamma_{t,N+1}$.  %In Section \ref{sec:simulation}, simulation results closely match the asymptotic analysis, i.e., Eq. (\ref{eq:P_f_HF_high_SNR}), at high SNR, which validates the provided FER expression.
We should emphasize that the asymptotic FER is more simple   and intuitive compared with the  FER expression in Eq.~(\ref{eq:P_f_HF_GRS}).

\section{Simulation results} \label{sec:simulation}
%---------------------------------------------------------------
%\begin{table}[!t]
%\centering
%\raggedleft
%\caption{SCENARIOS} \label{table:scenarios}
%\begin{tabular}{| c| c| c | c| c|}
%\hline
%        & System Type          & Number of Relays       & uncoded/coded   & SNRs \\  \hline \hline
%case 1  & No Relay System             & 0               & uncoded         & $\Omega_0=1$, $\Omega_{1i}=\Omega_{2i}=0$\\ \hline
%case 2  & No Relay System             & 0               & coded           & $\Omega_0=1$, $\Omega_{1i}=\Omega_{2i}=0$\\ \hline
%case 3  & hybrid DF/AF                & 3               & uncoded         & $\Omega_0=\Omega_{1i}=\Omega_{2i}=1$\\ \hline
%case 4  & hybrid DF/AF                & 3               & uncoded         & $\Omega_0=1$, $\Omega_{1i}=16$, $\Omega_{2i}=1$ \\ \hline
%case 5  & hybrid DF/AF                & 3               & uncoded         & $\Omega_0=1$, $\Omega_{1i}=1/16$, $\Omega_{2i}=1$\\ \hline
%case 6 & hybrid DF/AF                 & 3               & coded           & $\Omega_0=\Omega_{1i}=\Omega_{2i}=1$ \\ \hline
%case 7  & AF                          & 3               & uncoded            & $\Omega_0=1$,$\Omega_{1i}=16$, $\Omega_{2i}=1$  \\ \hline
%case 8  & AF                          & 3               & uncoded           & $\Omega_0=1$, $\Omega_{1i}=1/16$, $\Omega_{2i}=1$ \\ \hline
%case 9  & perfect DF                  & 3               & uncoded         & $\Omega_0=1$, $\Omega_{1i}=\infty$, $\Omega_{2i}=1$\\ \hline
%\end{tabular}
%\end{table}

In this section,  the analytical and the simulated results are provided for the SCHCN scheme.
Without specific mention, all simulations are performed by a BPSK modulation and a frame size
of $100$ symbols over block Rayleigh fading channels\cite{Huo2012P4998,Huo2013P5924,Huo2011P1}.
We consider the multiple relay nodes scenario, e.g., $N=1,2,3$.
Specifically, $N_c=0$ corresponds to the direct transmit case, and thus, no relay node is selected. $N_c=1$ denotes the best RS, and  all the relay nodes are participated for the transmission when  $N_c=N$.
Without specific mention,  we assume that a fixed total energy per symbol constraint in the network and equal power division among cooperating nodes, i.e., ($\mathcal{E}_s= \mathcal{E}_{r,i}=\frac{\mathcal{E}}{N_c+1}$), and $\lambda_{eq}=\lambda_{sr}+\lambda_{rd}$.
We also take into account the relay's location as: 1). the \emph{symmetric} case, where relays are placed halfway between the source and destination and 2). the \emph{asymmetric} case, where relays are close to the source or destination.
In Table~\ref{table:SNR threshold}, we show the SNR thresholds derived in Section \ref{sec:FER_approx_model} for the systems of different diversity order.
The simulation scenarios in this section are shown in Table~\ref{table:scenarios}.

%It is practically infeasible to verify the accuracy of analytical results for all the possible scenarios, as a prohibitively large number of combinations can be generated by varying relay number $n$ and average SNR parameters $\Omega$. Therefore, the scenarios in Table \ref{table:scenarios} are chosen such that a range of diverse cases are covered. For instance, case 2 is symmetric situations and case 3, 4 are asymmetric situations.
%We also choose case 1 for comparing the GHRS scheme to no-relay scheme. And in all the GRS scheme, $\gamma_{t,d}=\gamma_{t,N+1}$. The SNR threshold values are obtained using the methods developed in Section \ref{sec:FER_approx_model}, and are shown in Table \ref{table:SNR threshold}.

In Fig. \ref{fig:case.0}, we compare our proposed FER model  in Section \ref{sec:FER_approx_model} with the FER model in  \cite{Chatzigeorgiou2008P577} for case 0.
%shows the analytical FER curves by using the proposed model and the model of  \cite{Chatzigeorgiou2008P577} for case 0.
For case 0, the SNR thresholds are found to be $4.61$ dB based on the FER model in \cite{Chatzigeorgiou2008P577}, i.e., Eq. (\ref{eq.SNR.threshod.old}),   and are found to be $5.10$ dB, $5.36$ dB and $5.89$ dB for $N=1, 2, 4$, respectively,   based on our proposed FER model, i.e.,  Eq. (\ref{eq:SNR_threshold_continue}).
Fig. \ref{fig:case.0} shows that  our proposed model matches well with  the simulated FER as the SNR increases while the gap between the FER model in  \cite{Chatzigeorgiou2008P577} and the simulated FER cannot be ignored even at high SNR.
We also note that,  as  the diversity order increases, our proposed model is still accurate; however the FER results based on  the model in \cite{Chatzigeorgiou2008P577} become less accurate.%, but our proposed model is still accurate.
% While not shown here, similar trends can be observed for other scenarios.

In Fig. \ref{fig:SNR_threshold}, we present the optimal SNR threshold, i.e., Eq. (\ref{eq:opt.SNR.threshold}), obtained  by using the numerical method, when the FER  well matches the outage probability. We also include the proposed SNR threshold, i.e., Eq. (\ref{eq:SNR_threshold_continue}),  and the SNR threshold of \cite{Chatzigeorgiou2008P577},  i.e., Eq. (\ref{eq.SNR.threshod.old}). It can be observed that our proposed SNR threshold converges to the optimal SNR threshold at medium and high SNR regimes quickly, which  validates  our analysis in Section \ref{sec:FER_approx_model}.

In Fig. \ref{fig:Theory_Simulation_HF_3Relay_ClosedForm} and Fig. \ref{fig:Theory_Simulation_HF_3Relay_Simplified_individual_power}, we show the closed-form FER, asymptotic FER and simulated FER results of the SCHCN scheme for symmetric case, i.e., case 1.
From Fig. \ref{fig:Theory_Simulation_HF_3Relay_ClosedForm}, we can see that the closed-form curves given by Eq. (\ref{eq:P_f_HF_GRS}) match well with the simulated results  and  Fig. \ref{fig:Theory_Simulation_HF_3Relay_Simplified_individual_power} shows that the asymptotic results  given by Eq. (\ref{eq:P_f_HF_high_SNR}) converge to the simulated results at high SNRs.
These validate our  derived analytical FERs  in Section \ref{sec:FER_approx_model}.
In addition,  Fig. \ref{fig:Theory_Simulation_HDAF_3Relay_L_ClosedForm}   indicates that the derived FER of the SCHCN scheme is also accurate for the  case 2~(asymmetric case). Case 3 is also verified through simulation, but omitted here for brevity. Hence, our FER analysis is valid for both symmetric case and asymmetric case.
% and Fig. \ref{fig:Theory_Simulation_HDAF_3Relay_L_ClosedForm}

Fig. \ref{fig:Theory_Simulation_HF_3Relay_ClosedForm}   shows that  the best RS provides superior performance than the multiple RS under a total transmit power constraint, i.e., $\mathcal{E}_s= \mathcal{E}_{r,i}=\frac{\mathcal{E}}{N_c+1}$,   whereas Fig. \ref{fig:Theory_Simulation_HF_3Relay_Simplified_individual_power} shows that   the multiple RS outperforms the best RS under an individual transmit power constraint i.e., $\mathcal{E}_s= \mathcal{E}_{r,i}=\frac{\mathcal{E}}{N+1}$.  Moreover,
from Fig. \ref{fig:Theory_Simulation_HF_3Relay_Simplified_individual_power}  we note that the cooperative system with $N_c=2$ has very slight quality deterioration compared to the cooperative system with $N_c=3$. However,   the cooperative system with $N_c=2$ requires lower implementation complexity and   is more robust towards channel estimation errors. Hence, the SCHCN scheme has more merits than the conventional cooperative scheme in the practical systems for both with and without total transmit power constraint cases.
% and   Fig. \ref{fig:Theory_Simulation_HF_3Relay_Simplified}

%Fig. \ref{fig:Theory_Simulation_HF_3Relay_Simplified_bound} shows the derived lower bound and upper bound in Theorem \ref{theorem:upper-lower-bound}. From Table~\ref{table:SNR threshold}, note that the lower bound is close to the upper bound, the  approximate FER expression is tight when we use either the lower bound or the upper bound, which validates the analysis.

In Fig. \ref{fig:Theory_Simulation_HF_123Relay1_ClosedForm},  we compare   the SCHCN scheme with the non-cooperative system, i.e., $N_c=0$.
It is obviously shown that the SCHCN scheme can provide superior performance than the non-cooperative system, especially at high SNRs.
This is easy to understand as the cooperative system achieves a full diversity order of $N+1$ as long as $N_c>0$, whereas the non-cooperative system only has the diversity of $1$.
The simulation results of the SCHCN scheme with various frame lengths are shown in Fig.~\ref{fig:Theory_Simulation_HF_3Relay1_diff_L_ClosedForm}.
From Fig.~\ref{fig:Theory_Simulation_HF_3Relay1_diff_L_ClosedForm}, it shows that as the frame length decreases, the FER decreases. This finding is consistent with \cite{Lettieri1998P564}\footnote{We would refer the reader to \cite{Lettieri1998P564} for more details on the choice of a proper frame length.}.

\section{Conclusions} \label{sec:conclusion}
In this paper, we have analyzed the average FER of the SCHCN scheme in   cooperative wireless networks. Specifically, we considered a two-hop network with one source node, one destination node and multiple relay nodes. % with links experiencing independent block Rayleigh fading.
The closed-form  average FER expression as well as the   asymptotic FER expression at high SNRs have been derived for the SCHCN scheme.
Theoretical analysis  closely matches the simulated results.  %, and thus, validates our analysis.
Simulations also indicate that the best RS can achieve better performance than  the multiple RS under a total transmit power constraint, whereas the multiple RS outperforms the best RS  under an individual transmit power constraint.

%% Conclusion section
%% ==================================== %%

%% ==================================== %%
%% Appendix section

%% conference papers do not normally have an appendix
%% ----------------------
\appendices
%% ----------------------
%% ==================================== %%

%% ==================================== %%
%% Appendix 0
\section{Proof of Eq. (\ref{eq:ZECriterion})} \label{Proof:SNR_threshold}
%\begin{IEEEproof}%[Proof of Theorem \ref{thm:SVD}]
We notice that the
relative error at high SNR cannot be ignored because the integration in Eq.
(\ref{eq:RelativeCriterion}) is from $0$ to $\infty$. Hence,
we should have $\lim_{\bar{\gamma}
\rightarrow \infty}\left|
\frac{ \bar{P}_{f}(\bar{\gamma}) -
F_{\gamma}(\gamma_t,\bar{\gamma})}{\bar{P}_{f}(\bar{\gamma})}\right|\rightarrow 0$, which leads to Eq. (\ref{eq:ZECriterion}).

Otherwise,  suppose that existing a sufficiently big value $T$ ($0<T<\infty$) and a small enough value $\delta$ ($0< \delta < \infty$), the absolute relative error can be greater than $\delta$, i.e. $\left|
\frac{ \bar{P}_{f}(\bar{\gamma}) -
F_{\gamma}(\gamma_t,\bar{\gamma})}{\bar{P}_{f}(\bar{\gamma})}\right| >\delta, \text{when } \bar{\gamma}>T$.  Note that the sum of absolute relative error   cannot be minimized in this situation as it will approach infinity: $\int_{0}^{\infty} \left|
\frac{ \bar{P}_{f}(\bar{\gamma}) -
F_{\gamma}(\gamma_t,\bar{\gamma})}{\bar{P}_{f}(\bar{\gamma})}\right|
\text{d} \bar{\gamma} > \int_{T}^{\infty} \delta \text{d} \bar{\gamma} = \infty$.  Thus, Eq. (\ref{eq:ZECriterion}) is proved.
%\end{IEEEproof}

%% ----------------------

%% Appendix 0
%% ==================================== %%

%% ==================================== %%
%% Appendix I

%% ----------------------
\section{CDF of $\gamma_{i}$ in the SCHCN scheme: The Proof of Eq. (\ref{eq:CDF_HF_r_i}) and Eq. (\ref{eq:CDF_HF_r_i_Approximation})} \label{Proof:CDF_HF_SNR}

Note that the CDF of $\gamma_i$  is given as
\begin{equation} \begin{split}  \label{eq:end-to-end-i}
F_{\gamma_i}(\gamma ) &= Pr(\gamma_{i}<\gamma)\\
&= Pr(i \in \mathcal{G}_{AF})Pr(\gamma_{i}<\gamma|i \in \mathcal{G}_{AF})  \\
&\quad+Pr(i \in \mathcal{G}_{DF})Pr(\gamma_{i}<\gamma|i \in \mathcal{G}_{DF}).
\end{split} \end{equation}

In the HF scheme, the relays which decode correctly are included in  the  $\mathcal{G}_{DF}$  and the rest are  included in the $\mathcal{G}_{AF}$.  %the $i$-th relay node is divided into $\mathcal{G}_{DF}$ and $\mathcal{G}_{AF}$ based on whether it can decode correctly or not.
According to our proposed FER approximation model, a frame error only
occurs when SNR is below the SNR threshold $\gamma_t$. If $\gamma_{1i} \geq \gamma_{t,1}$ then $i \in \mathcal{G}_{DF}$, and if $\gamma_{1i} < \gamma_{t,1}$ then $i \in \mathcal{G}_{AF}$.
Hence, we have $Pr(i \in \mathcal{G}_{AF})=1-e^{-\lambda_{sr} \gamma_{t,1}}$ and  $Pr(i \in \mathcal{G}_{DF})=e^{-\lambda_{sr} \gamma_{t,1}}$, respectively.

If $i \in \mathcal{G}_{DF}$, $\gamma_i=\gamma_{2i}$, then the conditional CDF of
$\gamma_i$ given $i \in \mathcal{G}_{DF}$ can be expressed as
\begin{equation}  \label{eq:CDF_gamma_i_in_DF}
Pr(\gamma_{i}<\gamma|i \in \mathcal{G}_{DF}) = Pr(\gamma_{2i}<\gamma) =
1-e^{-\lambda_{rd} \gamma}.
\end{equation}

If $i \in \mathcal{G}_{AF} $, using the approximation $\frac{x y}{x+y+1} \approx min\{x,y\}$\cite{Anghel2004P1416}, then the corresponding conditional CDF of $\gamma_i$ given $i \in \mathcal{G}_{AF} $ can be expressed as \cite{David2003P}
%According to  \cite{Anghel2004P1416}, the corresponding conditional CDF of $\gamma_i$ given $i \in \mathcal{G}_{AF} $ can be written as\cite{David2003P}
\begin{equation} \begin{split} \label{eq:CDF_gamma_i_in_AF}
&\quad Pr(\gamma_{i}<\gamma|i \in \mathcal{G}_{AF})\\
&\approx Pr\left(min\{\gamma_{1i}|_{i \in \mathcal{G}_{AF}},\gamma_{2i}\} < \gamma\right) \\
&=1-(1-Pr(\gamma_{1i}<\gamma|i \in \mathcal{G}_{AF}))(1-Pr(\gamma_{2i}<\gamma) ),
\end{split} \end{equation}
where the conditional CDF of $\gamma_{1i}|i \in \mathcal{G}_{AF}$ is given by
\begin{equation}  \label{eq:CDF_gamma_1i_in_AF}
Pr(\gamma_{1i}<\gamma|i \in \mathcal{G}_{AF})=
\begin{cases}
\frac{ 1-e^{-\lambda_{sr} \gamma}}
{1-e^{-\lambda_{sr} \gamma_{t,1}}}, & \text{if } \gamma < \gamma_{t,1} \\
1.& \text{if } \gamma \geq \gamma_{t,1}
\end{cases}
\end{equation}

Combining Eq. (\ref{eq:CDF_gamma_i_in_AF}) and Eq. (\ref{eq:CDF_gamma_1i_in_AF}), the corresponding conditional CDF of $i \in \mathcal{G}_{AF} $ can be calculated
as
\begin{equation} \begin{split} \label{eq:CDF_gamma_i_in_AF2}
%F_{\gamma_i}(\gamma|i \in \mathcal{G}_{AF})
Pr(\gamma_{i}<\gamma|i \in \mathcal{G}_{AF})
& \approx
\begin{cases}
1-\frac{ e^{-\lambda_{sr} \gamma}-e^{-\lambda_{sr} \gamma_{t,1}}}
{1-e^{-\lambda_{sr} \gamma_{t,1}}} e^{-\lambda_{rd} \gamma}, & \text{if } \gamma < \gamma_{t,1}, \\
1,& \text{if } \gamma \geq \gamma_{t,1}.
\end{cases}
\end{split} \end{equation}

Thus, after doing some substitutions in Eq. (\ref{eq:end-to-end-i}), it yields Eq. (\ref{eq:CDF_HF_r_i}).
When $\gamma_{t,1} \leq \gamma \leq \gamma_{t,d}$, we use an exponential function to approximate Eq. (\ref{eq:CDF_HF_r_i}).
Let $1-e^{-\lambda_{sr} \gamma_{t,1}}e^{-\lambda_{rd}\gamma} = 1 - e^{-\lambda_{eq}\gamma}$, that is, $\lambda_{eq} =\lambda_{sr} \gamma_{t,1}/\gamma +\lambda_{rd}$. Hence, we have
$\lambda_{sr} \gamma_{t,1}/\gamma_{t,d} +\lambda_{rd} \leq \lambda_{eq} \leq \lambda_{sr} +\lambda_{rd}$. We note that this inequality  still holds when considering  the whole range from $0$ to $\gamma_{t,d}$.

%% ----------------------

%% Appendix I
%% ==================================== %%

%% ==================================== %%
%% Appendix II

%% ----------------------
\section{Proof of Property \ref{theorem:math_PDF}} \label{Proof:CDF_SNR}% Theorem

According to Eq. (\ref{eq:SNR_GRS_MRC}), $\gamma_{SCHCN} =\gamma_{0}+\gamma_{coop}$, in the following, we first derive the Laplace transforms of the PDF of $\gamma_{coop}$ and $\gamma_{0}$, and then the PDF of $\gamma_{SCHCN}$ can be directly obtained by using the inverse Laplace transform.

As $\{\gamma_{i}\}_{i=1}^{N}$   are i.i.d. exponential RVs , with the help of \cite[Eq. (9.321)]{Simon2005P}, the Laplace transform of the PDF of $\gamma_{coop}$ can be expressed as
\begin{equation} \begin{split}
\mathcal{L}_{\gamma_{coop}}(s)
%&=\frac{1}{(1+s\overline\gamma_{eq})^{N_c-1} \prod_{n=N_c}^N(1+\frac{s\overline\gamma_{eq} N_c}{n})}\\
&=\frac{1}{(1+\frac{s}{\lambda_{eq}})^{N_c-1} \prod_{n=N_c}^N(1+\frac{s N_c}{\lambda_{eq} n})}\\
%&=\frac{\frac{N!}{N_c!N_c^{N-N_c}}}{(1+\frac{s}{\lambda_{eq}})^{N_c} \prod_{n=1}^{N-N_c}(1+\frac{n}{N_c}+\frac{s}{\lambda_{eq}})}\\
&=\frac{\frac{N!}{N_c!N_c^{N-N_c}}\lambda_{eq}^N}{(s+\lambda_{eq})^{N_c} \prod_{n=1}^{N-N_c}[s+(1+\frac{n}{N_c})\lambda_{eq}]},
\end{split} \end{equation}
and with the help of \cite[Eq. (17.13.7)]{Gradshteyn2007P}, the Laplace transform of the PDF of $\gamma_{0}$ can be expressed as
\begin{equation} \begin{split}
\mathcal{L}_{\gamma_{0}}(s)
%&=\frac{1}{1+ s\overline\gamma_{0}}\\
%&=\frac{1}{1+\frac{s}{\lambda_{0}}}\\
&=\frac{\lambda_{0}}{s+\lambda_{0}}.
\end{split} \end{equation}
 .

As $\gamma_{coop}$ and $\gamma_{0}$ are mutually independent, the Laplace transforms of the PDF of $\gamma_{SCHCN}$ is given as
\begin{equation} \begin{split} \label{eq:prod_Laplace_gamma_GRS_case1}
 \mathcal{L}_{\gamma_{SCHCN}}(s)
&= \mathcal{L}_{\gamma_{coop}}(s)\cdot \mathcal{L}_{\gamma_{0}}(s) \\
&=\frac{\frac{N!}{N_c!N_c^{N-N_c}}\lambda_{0}\lambda_{eq}^N}{ (s+\lambda_{eq})^{N_c} (s+\lambda_{0}) \prod_{n=1}^{N-N_c}[s+(1+\frac{n}{N_c})\lambda_{eq}]}.
\end{split} \end{equation}

Note that Eq. (\ref{eq:prod_Laplace_gamma_GRS_case1}) can be rewritten for the two cases $\lambda_{0}\neq \lambda_{eq}$ and $\lambda_{0}=\lambda_{eq}$ as follows.

\subsubsection{Case of $\lambda_{0}\neq \lambda_{eq}$}
Applying the partial fraction expansion\cite{Latni1998P}, Eq. (\ref{eq:prod_Laplace_gamma_GRS_case1}) can be rewritten as
\begin{equation} \begin{split}\label{eq:sum_Laplace_gamma_GRS_case1}
\mathcal{L}_{\gamma_{SCHCN}}(s)&=\sum_{i=1}^{N_c} \frac{\alpha_i}{(s+\lambda_{eq})^{i}}+\frac{\beta_0}{s+\lambda_{0}}+\sum_{j=1}^{N-N_c}\frac{\beta_j}{s+(1+\frac{j}{N_c})\lambda_{eq}},\\
\end{split} \end{equation}
where
\begin{equation} \begin{split}
c=\frac{N!}{N_c!N_c^{N-N_c}}\lambda_{0}\lambda_{eq}^N,
\end{split} \end{equation}
\begin{equation} \begin{split}
\theta_0=\frac{1}{{\prod_{n=1}^{N-N_c} [(1+\frac{n}{N_c})\lambda_{eq}-\lambda_{0}]}} ,
\end{split} \end{equation}
\begin{equation} \begin{split}
\theta_k&=\frac{1}{[-(1+\frac{k}{N_c})\lambda_{eq}+\lambda_{0}]\prod_{n=1,n \neq k}^{N-N_c} [(\frac{n-k}{N_c})\lambda_{eq}]},% \\
%&=\frac{1}{ [-(1+\frac{k}{N_c})\lambda_{eq}+\lambda_{0}](-1)^{k-1}(k-1)!(N-N_c-k)! (\frac{\lambda_{eq}}{N_c})^{N-N_c-1} },
\end{split} \end{equation}
\begin{equation} \begin{split}
\alpha_i
&=\frac{c}{(N_c-i)!}\frac{d^{N_c-i}}{ds^{N_c-i}}\left[\frac{1}{(s+\lambda_{0})\prod_{n=1}^{N-N_c} [s+(1+\frac{n}{N_c})\lambda_{eq}]} \right] \bigg|_{s=-\lambda_{eq}}\\
%&=\frac{c}{(N_c-i)!}\frac{d^{N_c-i}}{ds^{N_c-i}}\left[\frac{\theta_0}{s+\lambda_{0}} +\sum_{k=1}^{N-N_c}\frac{\theta_k}{s+(1+\frac{k}{N_c})\lambda_{eq}} \right] \bigg|_{s=-\lambda_{eq}}\\
%&=\frac{c}{(N_c-i)!} \left[\frac{\theta_0(-1)^{N_c-i}(N_c-i)!}{(s+\lambda_{0})^{N_c-i+1}} +\sum_{k=1}^{N-N_c} \frac{\theta_k(-1)^{N_c-i}(N_c-i)!}{( s+(1+\frac{k}{N_c}) \lambda_{eq})^{N_c-i+1}} \right] \bigg|_{s=-\lambda_{eq}}\\
&=\frac{c}{(N_c-i)!} \left[\frac{\theta_0(-1)^{N_c-i}(N_c-i)!}{(-\lambda_{eq}+\lambda_{0})^{N_c-i+1}} +\sum_{k=1}^{N-N_c} \frac{\theta_k(-1)^{N_c-i}(N_c-i)!}{( \frac{k}{N_c} \lambda_{eq})^{N_c-i+1}} \right],
\end{split} \end{equation}
\begin{equation} \begin{split}
\beta_0&=\frac{c}{(\lambda_{eq}-\lambda_{0})^{N_c}\prod_{n=1}^{N-N_c}[(1+\frac{n}{N_c})\lambda_{eq}-\lambda_{0}]},\\
%&=\frac{c\cdot \theta_0}{(\lambda_{eq}-\lambda_{0})^{N_c} },\\
\end{split} \end{equation}
and
\begin{equation} \begin{split}
\beta_j&=\frac{c}{(-\frac{j}{N_c}\lambda_{eq})^{N_c} [-(1+\frac{j}{N_c})\lambda_{eq}+\lambda_{0}] \prod_{n=1,n\neq j}^{N-N_c}[(\frac{n-j}{N_c})\lambda_{eq}]}.\\
%&=\frac{c}{(-\frac{j}{N_c}\lambda_{eq})^{N_c} [-(1+\frac{j}{N_c})\lambda_{eq}+\lambda_{0}] (-1)^{j-1}(j-1)! (N-N_c-j)!(\frac{\lambda_{eq}}{N_c})^{N-N_c-1}}\\
%&=\frac{c\cdot \theta_j}{ (-\frac{j}{N_c}\lambda_{eq})^{N_c} }.
\end{split} \end{equation}

Applying the inverse Laplace transforms in Eq. (\ref{eq:sum_Laplace_gamma_GRS_case1}), we can obtain the PDF of $\gamma_{SCHCN}$ in closed form as \cite[Eq. (17.13.17)]{Gradshteyn2007P}
\begin{equation} \begin{split}
f_{\gamma_{SCHCN}}(\gamma)&=\mathcal{L}^{-1}\left[ \mathcal{L}_{\gamma_{SCHCN}}(s)\right]\\
%&=\mathcal{L}^{-1}\left[ \sum_{i=1}^{N_c} \frac{\alpha_i}{(s+\lambda_{eq})^{i}}+\frac{\beta_0}{s+\lambda_{0}} +\sum_{j=1}^{N-N_c}\frac{\beta_j}{s+(1+\frac{j}{N_c})\lambda_{eq}}\right]\\
&=\sum_{i=1}^{N_c} \frac{ \alpha_i }{(i-1)!}\gamma^{i-1} e^{-\lambda_{eq}\gamma}+ \beta_0 e^{-\lambda_{0}\gamma} +\sum_{j=1}^{N-N_c} \beta_j e^{-(1+\frac{j}{N_c})\lambda_{eq}\gamma}.\\
\end{split} \end{equation}

The CDF of $\gamma_{SCHCN}$ can be obtained directly by integrating the PDF of $\gamma_{SCHCN}$ in closed form as
\begin{equation} \begin{split}
F_{\gamma_{SCHCN}}(\gamma) &= \int_0^\gamma f_{\gamma_{SCHCN}}(t)\text{d}t\\
&=\sum_{i=1}^{N_c} \frac{ \alpha_i }{(i-1)!} g(i-1,\lambda_{eq},\gamma)
+ \beta_0 g(0,\lambda_{0},\gamma) +\sum_{j=1}^{N-N_c} \beta_j\cdot g\left(0,(1+\frac{j}{N_c})\lambda_{eq},\gamma\right) \\
&=\sum_{i=1}^{N_c} \frac{ \alpha_i } {\lambda_{eq} ^{i}}\left[ 1-e^{-\lambda_{eq} \gamma}\sum_{m=0}^{i-1} \frac{(\lambda_{eq} \gamma)^m}{m!}\right] +\sum_{j=1}^{N-N_c} \frac{\beta_j}{(1+\frac{j}{N_c})\lambda_{eq}}(1-e^{-(1+\frac{j}{N_c})\lambda_{eq}\gamma}) \\
&\qquad \qquad+ \frac{\beta_0}{\lambda_0}(1-e^{-\lambda_{0}\gamma}),
\end{split} \end{equation}
where   $g(n,\beta,\gamma)$ is defined as \cite[Eq. (3.351.1)]{Gradshteyn2007P}
\begin{equation} \begin{split}
g(n,\beta,\gamma)&=\int_0^x e^{-\beta t}t^n\text{d}t\\%=\frac{\gamma^*(n+1,\beta x)}{\beta^{n+1}}\\
& =\frac{n!}{\beta ^{n+1}}\left[ 1-e^{-\beta x}\sum_{m=0}^n \frac{(\beta x)^m}{m!}\right] ,\\
&\quad [\gamma>0,Re\,\beta >0,n=0,1,2,...].
\end{split} \end{equation}

%%****************************************************************
\subsubsection{Case of $\lambda_{0}=\lambda_{eq}$}
Eq. (\ref{eq:prod_Laplace_gamma_GRS_case1}) can be rewritten as

\begin{equation} \begin{split} \label{eq:prod_Laplace_gamma_GRS_case2}
\mathcal{L}_{\gamma_{SCHCN}}(s)&=\mathcal{L}_{\gamma_{coop}}(s)\cdot \mathcal{L}_{\gamma_{0}}(s)\\
&=\frac{\frac{N!}{N_c!N_c^{N-N_c}}\lambda_{0}^{N+1}}{(s+\lambda_{0})^{N_c+1} \prod_{n=1}^{N-N_c}[s+(1+\frac{n}{N_c})\lambda_{0}]}.
\end{split} \end{equation}
Hence, applying the partial fraction expansion in Eq. (\ref{eq:prod_Laplace_gamma_GRS_case2}), we can further simplify it as
\begin{equation} \begin{split}\label{eq:sum_Laplace_gamma_GRS_case2}
\mathcal{L}_{\gamma_{SCHCN}}(s)&=\sum_{i=1}^{N_c+1} \frac{\alpha_i^*}{(s+\lambda_{0})^{i}} +\sum_{j=1}^{N-N_c}\frac{\beta_j^*}{s+(1+\frac{j}{N_c})\lambda_{0}},\\
\end{split} \end{equation}
where
\begin{equation} \begin{split}
c^*=\frac{N!}{N_c!N_c^{N-N_c}}\lambda_{0}^{N+1},
\end{split} \end{equation}
\begin{equation} \begin{split}
\theta_k^*&=\frac{1}{  \prod_{n=1,n \neq k}^{N-N_c} [(\frac{n-k}{N_c})\lambda_{0}]}, \\
%&=\frac{1}{ (-1)^{k-1}(k-1)!(N-N_c-k)! (\frac{\lambda_{0}}{N_c})^{N-N_c-1} },
\end{split} \end{equation}
\begin{equation} \begin{split}
\alpha_i^*
&=\frac{c^*}{(N_c+1-i)!}\frac{d^{N_c+1-i}}{ds^{N_c+1-i}}\left[\frac{1}{ \prod_{n=1}^{N-N_c} [s+(1+\frac{n}{N_c})\lambda_{0}]} \right] \bigg|_{s=-\lambda_{0}}\\
%&=\frac{c^*}{(N_c+1-i)!}\frac{d^{N_c+1-i}}{ds^{N_c+1-i}}\left[\sum_{k=1}^{N-N_c}\frac{\theta_k^*}{s+(1+\frac{k}{N_c})\lambda_{eq}} \right] \bigg|_{s=-\lambda_{0}}\\
%&=\frac{c^*}{(N_c+1-i)!} \left[ \sum_{k=1}^{N-N_c} \frac{\theta_k^*(-1)^{N_c+1-i}(N_c+1-i)!}{( s+(1+\frac{k}{N_c}) \lambda_{eq})^{N_c-i+2}} \right] \bigg|_{s=-\lambda_{0}}\\
&=\frac{c^*}{(N_c+1-i)!} \left[ \sum_{k=1}^{N-N_c} \frac{\theta_k^*(-1)^{N_c+1-i}(N_c+1-i)!}{( \frac{k}{N_c} \lambda_{0})^{N_c-i+2}} \right],
\end{split} \end{equation}
and
\begin{equation} \begin{split}
\beta_j^*&=\frac{c^*}{(-\frac{j}{N_c}\lambda_{0})^{N_c+1}   \prod_{n=1,n\neq j}^{N-N_c}[(\frac{n-j}{N_c})\lambda_{0}]}.\\
%&=\frac{c^*}{(-\frac{j}{N_c}\lambda_{0})^{N_c+1} (-1)^{j-1}(j-1)! (N-N_c-j)!(\frac{\lambda_{0}}{N_c})^{N-N_c-1}}\\
%&=\frac{c^*\cdot\theta_j^*}{ (-\frac{j}{N_c}\lambda_{0})^{N_c+1} }.
\end{split} \end{equation}

Similarly, we can derive the expression of the PDF and CDF of $\gamma_{SCHCN}$ as follows:
\begin{equation} \begin{split}
f_{\gamma_{SCHCN}}(\gamma)&=\mathcal{L}^{-1}\left[ \mathcal{L}_{\gamma_{SCHCN}}(s)\right]=\sum_{i=1}^{N_c+1} \frac{ \alpha_i^* }{(i-1)!}\gamma^{i-1} e^{-\lambda_{0}\gamma}+ \sum_{j=1}^{N-N_c} \beta_j^* e^{-(1+\frac{j}{N_c})\lambda_{0}\gamma},\\
\end{split} \end{equation}
and
\begin{equation} \begin{split}
F_{\gamma_{SCHCN}}(\gamma) &= \int_0^\gamma f_{\gamma_{SCHCN}}(t)\text{d}t\\
&=\sum_{i=1}^{N_c+1} \frac{ \alpha_i^* } {\lambda_{eq} ^{i}}\left[ 1-e^{-\lambda_{0} \gamma}\sum_{m=0}^{i-1} \frac{(\lambda_{0} \gamma)^m}{m!}\right] +\sum_{j=1}^{N-N_c} \frac{\beta_j^*}{(1+\frac{j}{N_c})\lambda_{0}}(1-e^{-(1+\frac{j}{N_c})\lambda_{0}\gamma}). \\
\end{split} \end{equation}
Hence, Property \ref{theorem:math_PDF} is proved. %It should mention here that this property can also be proved when using the order statistics theory in\cite[Chapter.~11]{Balakrishnan2006P}.

%% ----------------------

%% Appendix II
%% ==================================== %%

%% ==================================== %%
%% Appendix III

%% ----------------------
\section{Proof of Property \ref{theorem:math_PDF_simplified}} \label{Proof:CDF_SNR_at_high_SNR}
The Laplace transform of the PDF of $\gamma_{SCHCN}$ is derived in Appendix \ref{Proof:CDF_SNR} as
\begin{equation} \begin{split}
\mathcal{L}_{\gamma_{SCHCN}}(s)&=\frac{\frac{N!}{N_c!N_c^{N-N_c}}\lambda_{0}\lambda_{eq}^N}{(s+\lambda_{eq})^{N_c} (s+\lambda_{0}) \prod_{n=1}^{N-N_c}[s+(1+\frac{n}{N_c})\lambda_{eq}]}.
\end{split} \end{equation}
Since the Laplace transform of the PDF and CDF of $\gamma_{SCHCN}$ are related by
\begin{equation} \begin{split}
\hat{\mathcal{L}}_{\gamma_{SCHCN}}(s)= \frac{\mathcal{L}_{\gamma_{SCHCN}}(s)}{s},
\end{split} \end{equation}
where $\hat{\mathcal{L}}_{\gamma_{SCHCN}}(s)$ is the Laplace transform of the CDF of $\gamma_{SCHCN}$ .
Applying the inverse Laplace transforms, the CDF of $\gamma_{SCHCN}$ can be expressed as
\begin{equation} \begin{split} \label{eq:CDF_SNR_GRS_Simple0}
F_{\gamma_{SCHCN}}(\gamma)&=\mathcal{L}^{-1}\left[\frac{\mathcal{L}_{\gamma_{SCHCN}}(s)}{s}\right]\\
&=\mathcal{L}^{-1}\left[\frac{\frac{N!}{N_c!N_c^{N-N_c}}\lambda_{0}\lambda_{eq}^N}{s(s+\lambda_{eq})^{N_c} (s+\lambda_{0}) \prod_{n=1}^{N-N_c}[s+(1+\frac{n}{N_c})\lambda_{eq}]}\right].
\end{split} \end{equation}

At high SNR, $\overline\gamma_{0}$ and $\overline\gamma_{eq}$ are large. According to $\lambda_0=1/\overline\gamma_{0}$ and $\lambda_{eq}=1/\overline\gamma_{eq}$, we know that $\lambda_0$ and $\lambda_{eq}$ are small.
Hence, Eq. (\ref{eq:CDF_SNR_GRS_Simple0}) can be simplified as\cite[Eq. (17.13.2)]{Gradshteyn2007P}
\begin{equation} \begin{split}
F_{\gamma_{SCHCN}}(\gamma)& \approx \mathcal{L}^{-1}\left[\frac{N!}{N_c!{N_c}^{N-N_c}}\frac{\lambda_0{\lambda_{eq}}^N}{s^{N+2}}\right]\\
&=\frac{1}{(N+1)N_c!{N_c}^{N-N_c}} \lambda_{0}{\lambda_{eq}}^N\gamma^{N+1}.
\end{split} \end{equation}

Hence, Eq. (\ref{eq:CDF_high_SNR}) is proved.
%% ----------------------

%% Appendix III
%% ==================================== %%

%% ==================================== %%
%% Acknowledgment section

%% ----------------------
%\section*{Acknowledgment}
%This work was partially supported by the National Natural Science
%Foundation of China under Grant number 60972009, National Science
%and Technology Major Projects of China under Grant number
%2009ZX03003-011, 2010ZX03005-003, 2009ZX03003-001, and by Shanghai
%Leading Aca- demic Discipline Project under project Number T0102.
%% ----------------------

%% Acknowledgment section
%% ==================================== %%

%% ==================================== %%
%% reference section

%% ----------------------
%\newpage
%\renewcommand{\baselinestretch}{1.2}   % �޸Ĳο������о�
%% ----------------------

\bibliographystyle{ieeetran}
%\bibliography{IEEEabrv,GHRS}
%\bibliography{IEEEabrv,E:/imhuo/Dropbox/Study/Template/Paper/Ref/JabRef/Using/Huo}

%% ----------------------
%\bibliographystyle{IEEEtranSN}%SN = sorted, plainnat.bst(using name and year)
%\bibliographystyle{IEEEtran}
%\bibliography{IEEEabrv,mybibfile}
%% ----------------------

%% ----------------------

%% ----------------------

%% reference section
%% ==================================== %%

%% ==================================== %%
%% table

%% ----------------------
\newpage
%% ----------------------

%---------------------------------------------------------------
%\begin{table*}[!t]
%\centering
%%\raggedleft
%\caption{SNR THRESHOLD} \label{table:SNR threshold}
%\begin{tabular}{| c| c| c| c |}
%\hline
%Number of Relays: $N$ & Diversity Order: $d$ & Proposed SNR Threshold: $\gamma_{t,d}$ & SNR Threshold of \cite{Chatzigeorgiou2008P577}: $\gamma_{t}$ \\ \hline \hline
%0 & 1 & 5.10 dB & \\ \cline{1-3} %\hline
%1 & 2 & 5.36 dB & 4.61 dB\\ \cline{1-3} %\hline
%2 & 3 & 5.62 dB & \\ \cline{1-3} %\hline
%3 & 4 & 5.89 dB & \\ \hline
%\end{tabular}
%\end{table*}

\begin{table*}[!t]
\centering
%\raggedleft
\caption{SNR THRESHOLD} \label{table:SNR threshold}
\begin{tabular}{| c| c| c| c |}
\hline
  Diversity Order: $d$ & Frame Length: $L$ & Proposed SNR Threshold: $\gamma_{t,d}$, dB & SNR Threshold of \cite{Chatzigeorgiou2008P577}: $\gamma_{t}$, dB \\ \hline \hline
  1 & 100& 5.10  & \\ \cline{1-3} %\hline
  2 & 100& 5.36  & 4.61 \\ \cline{1-3} %\hline
  3 & 100& 5.62  & \\ \cline{1-3} %\hline
  4 & 100& 5.89  & \\ \hline
  4 & 200& 6.45  & 5.50  \\ \hline  %\hline
  4 & 400& 6.97  & 6.24  \\ \hline %\hline
\end{tabular}
\end{table*}
%---------------------------------------------------------------

%---------------------------------------------------------------
%\begin{table*}[!t]
%\centering
%%\raggedleft
%\caption{SCENARIOS} \label{table:scenarios}
%\begin{tabular}{| c| c| c | c| c|}
%\hline
%        & System Type & Number of Nodes & SNRs \\ \hline \hline
%case 0  & MIMO        & $N_T=1, N=N_R=1,2,4$            & $\Omega =1$ \\ \hline
%case 1  & GHRS        & $N=3$               & $\Omega_0=\Omega_{1i}=\Omega_{2i}=1$\\ \hline
%case 2  & GHRS        & $N=3$               & $\Omega_0=1$, $\Omega_{1i}=16$, $\Omega_{2i}=1$ \\ \hline
%case 3  & GHRS        & $N=3$                & $\Omega_0=1$, $\Omega_{1i}=1/16$, $\Omega_{2i}=1$\\ \hline
%\end{tabular}
%\end{table*}

\begin{table*}[!t]
\centering
%\raggedleft
\caption{SCENARIOS} \label{table:scenarios}
\begin{tabular}{| c| c| c| c | c| c|}
\hline
        & System Type & Number of Nodes & Diversity Order: $d$ & SNRs \\ \hline \hline
case 0  & MIMO        & $N_T=1, N=N_R=1,2,4$    &$d=1,2,4$        & $\Omega =1$ \\ \hline
case 1  & SCHCN        & $N=1,2,3$     &$d=2,3,4$          & $\Omega_0=\Omega_{1i}=\Omega_{2i}=1$\\ \hline
case 2  & SCHCN        & $N=3$    &$d=4$            & $\Omega_0=1$, $\Omega_{1i}=16$, $\Omega_{2i}=1$ \\ \hline
case 3  & SCHCN        & $N=3$   &$d=4$              & $\Omega_0=1$, $\Omega_{1i}=1/16$, $\Omega_{2i}=1$\\ \hline
\end{tabular}
\end{table*}
%---------------------------------------------------------------

%% table
%% ==================================== %%

%% ==================================== %%
%% figure

%% ----------------------
\newpage
%% ----------------------

%% System Diagram
%%--------------------------------------------------
\begin{figure}%[!t]%[!b]
\centering
%\graphicspath{{}}
\includegraphics[width=0.7\textwidth]{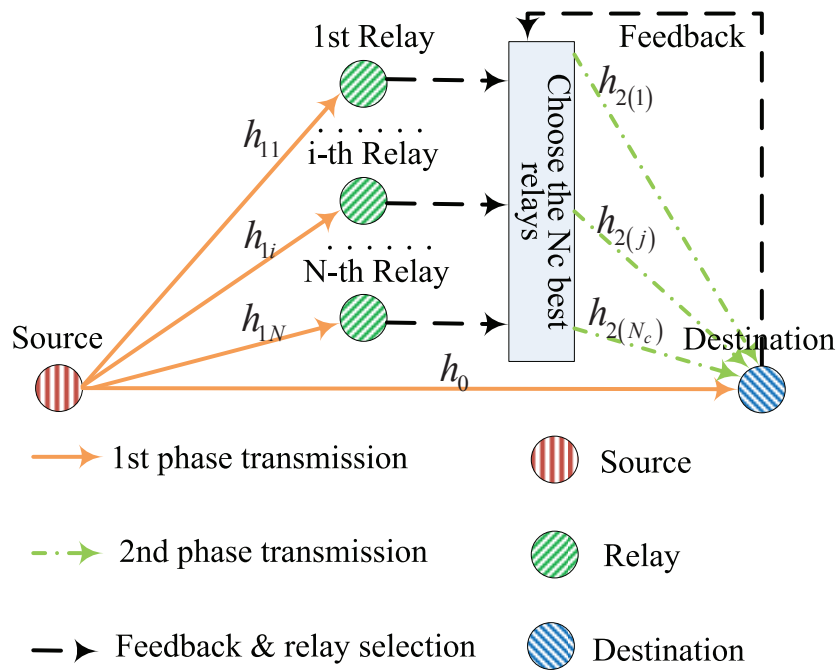}
%\caption{Diagram of the cooperative wireless system with the SCHCN scheme.}
\caption{Block diagram of the  SCHCN scheme.} \label{fig:SystemDiagram_GRS}
\end{figure}

%%---------------------------------------------------------------
%\clearpage
\begin{figure}[!t]
\centering
\graphicspath{{}}
\includegraphics[width=0.7\textwidth]{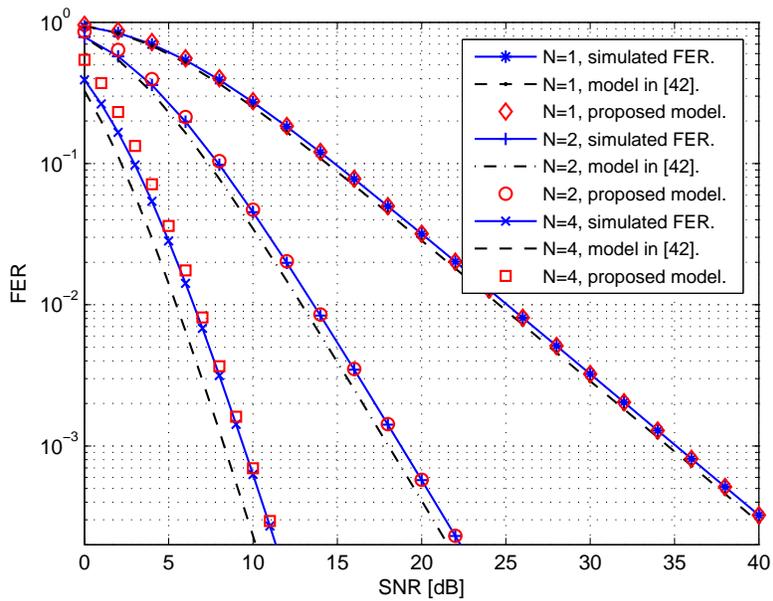}
\caption{FER comparison of the proposed model and the model of  \cite{Chatzigeorgiou2008P577} for case 0: general MIMO channels with $N_T=1$, $N=N_R=1,2,4$, $\Omega =1$ and $L=100$.}%uncoded.
 \label{fig:case.0}
\end{figure}

%---------------------------------------------------------------
%\clearpage
\begin{figure}[!t]
\centering
%\graphicspath{{}}
\includegraphics[width=0.7\textwidth]{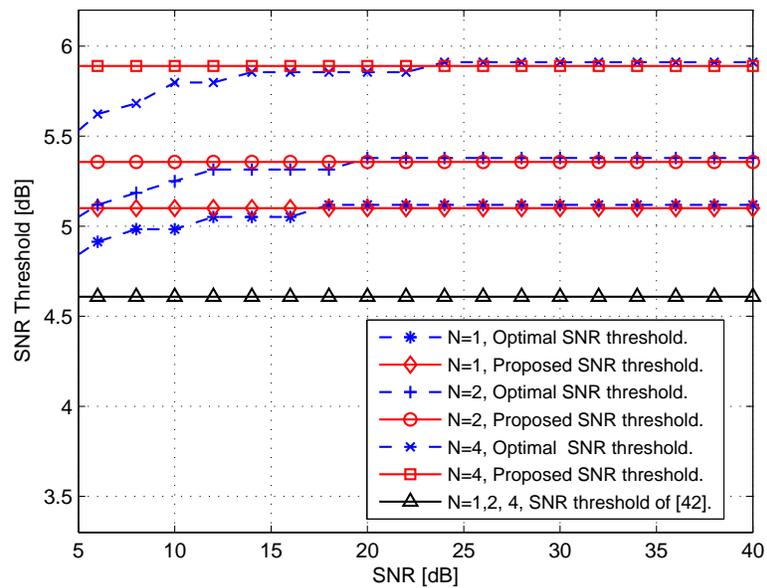}
\caption{SNR threshold comparison of the proposed model and the model of  \cite{Chatzigeorgiou2008P577} for case 0: general MIMO channels with $N_T=1$, $N=N_R=1,2,4$, $\Omega =1$ and $L=100$.}%uncoded.
\label{fig:SNR_threshold}
\end{figure}

%%****************************************************************
%% Simulation
\begin{figure}[!t]%[htbp]
\centering
%\graphicspath{{}}
\includegraphics[width=0.7\textwidth]{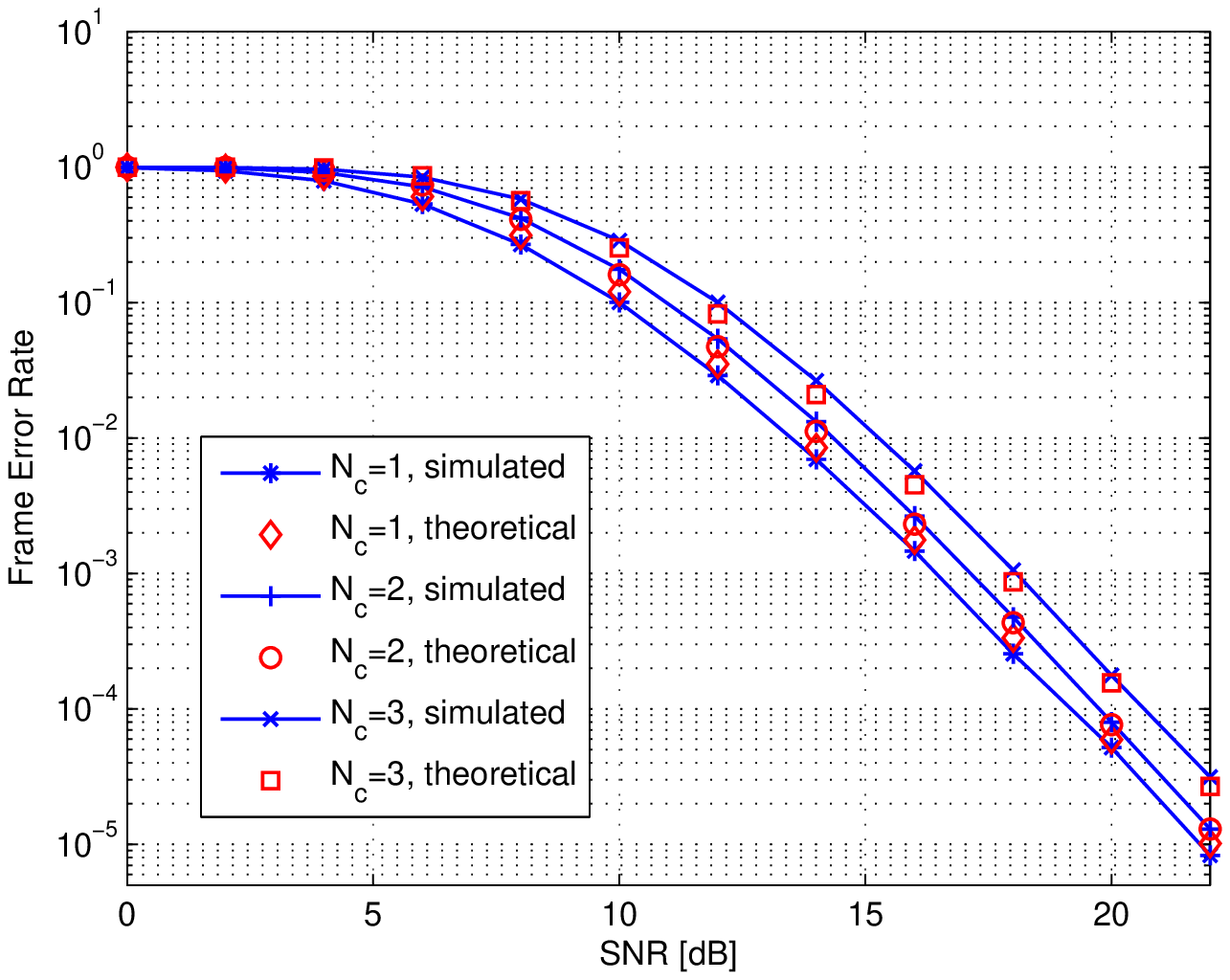}
\caption{Closed-form and simulated FER results of the SCHCN scheme under a total transmit power constraint: $N=3$, $\Omega_0=\Omega_{1i}=\Omega_{2i}=1$ and $L=100$.
}  \label{fig:Theory_Simulation_HF_3Relay_ClosedForm}
\end{figure}

%%****************************************************************
%% Simulation
\begin{figure}[!t]%[htbp]
\centering
%\graphicspath{{}}
\includegraphics[width=0.7\textwidth]{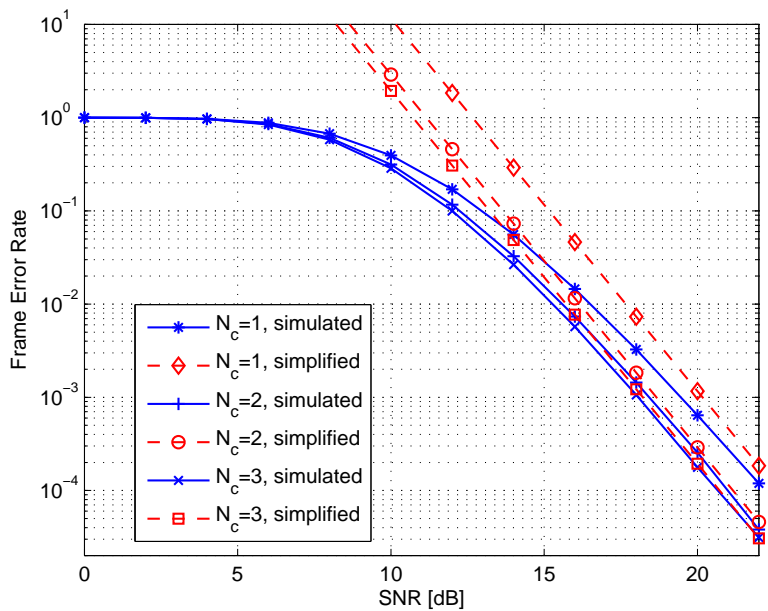}
\caption{Asymptotic and simulated FER results of the SCHCN scheme under an individual power constraint: $N=3$, $\Omega_0=\Omega_{1i}=\Omega_{2i}=1$ and $L=100$.
}  \label{fig:Theory_Simulation_HF_3Relay_Simplified_individual_power}
\end{figure}

%%****************************************************************
%%% Simulation
%\begin{figure}[!t]%[htbp]
%\centering
%%\graphicspath{{}}
%\includegraphics[width=1\textwidth]{Theory_Simulation_HDAF_3Relay_Simplified_v1_0}
%\caption{Simplified and simulated FER results of the GHRS scheme: $N=3$.}  \label{fig:Theory_Simulation_HF_3Relay_Simplified}
%\end{figure}
%%****************************************************************
%%% Simulation
%\begin{figure}[!t]%[htbp]
%\centering
%%\graphicspath{{}}
%\includegraphics[width=1\textwidth]{Theory_Simulation_HDAF_3Relay_R_ClosedForm_v1_0}
%\caption{Closed-form and simulated FER results of the GHRS scheme: $N=3$, $\Omega_0=1$, $\Omega_{1i}=\frac{1}{16}$, $\Omega_{2i}=1$.
%}  \label{fig:Theory_Simulation_HDAF_3Relay_R_ClosedForm}
%\end{figure}

%%****************************************************************
%% Simulation
\begin{figure}[!t]%[htbp]
\centering
%\graphicspath{{}}
\includegraphics[width=0.7\textwidth]{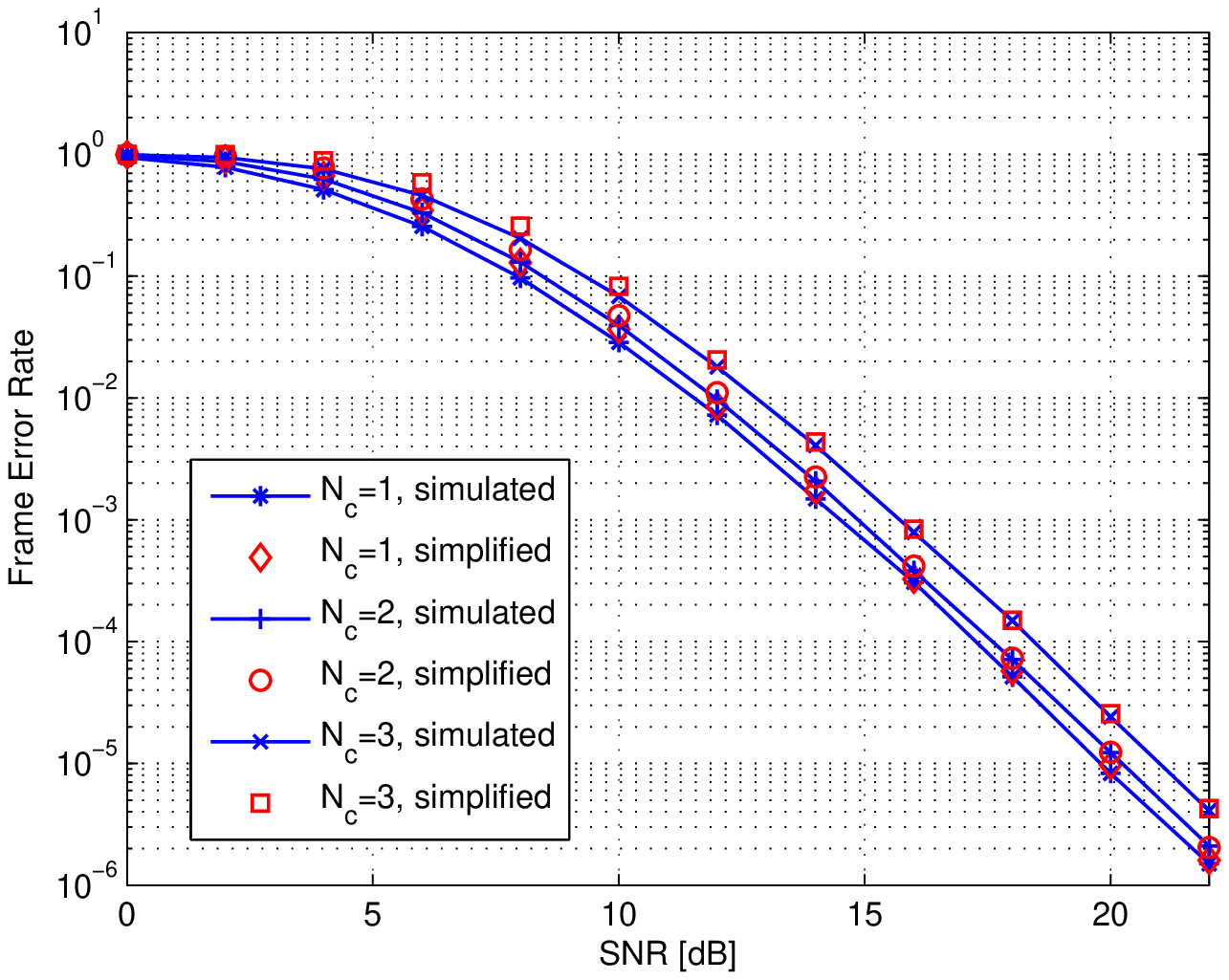}
\caption{Closed-form and simulated FER results of the SCHCN scheme under a total transmit power constraint: $N=3$, $\Omega_0=1$, $\Omega_{1i}=16$, $\Omega_{2i}=1$ and $L=100$.
}  \label{fig:Theory_Simulation_HDAF_3Relay_L_ClosedForm}
\end{figure}

 %%****************************************************************
%%% Simulation
%\begin{figure}[!t]%[htbp]
%\centering
%%\graphicspath{{}}
%\includegraphics[width=1\textwidth]{Theory_Simulation_HDAF_3Relay_Simplified_bound_v1_0}
%\caption{Lower-bound based simplified, upper-bound based simplified and simulated FER results of the GHRS scheme: $N=3$.
%}  \label{fig:Theory_Simulation_HF_3Relay_Simplified_bound}
%\end{figure}

%%****************************************************************
%% Simulation
\begin{figure}[!t]%[htbp]
\centering
%\graphicspath{{}}
\includegraphics[width=0.7\textwidth]{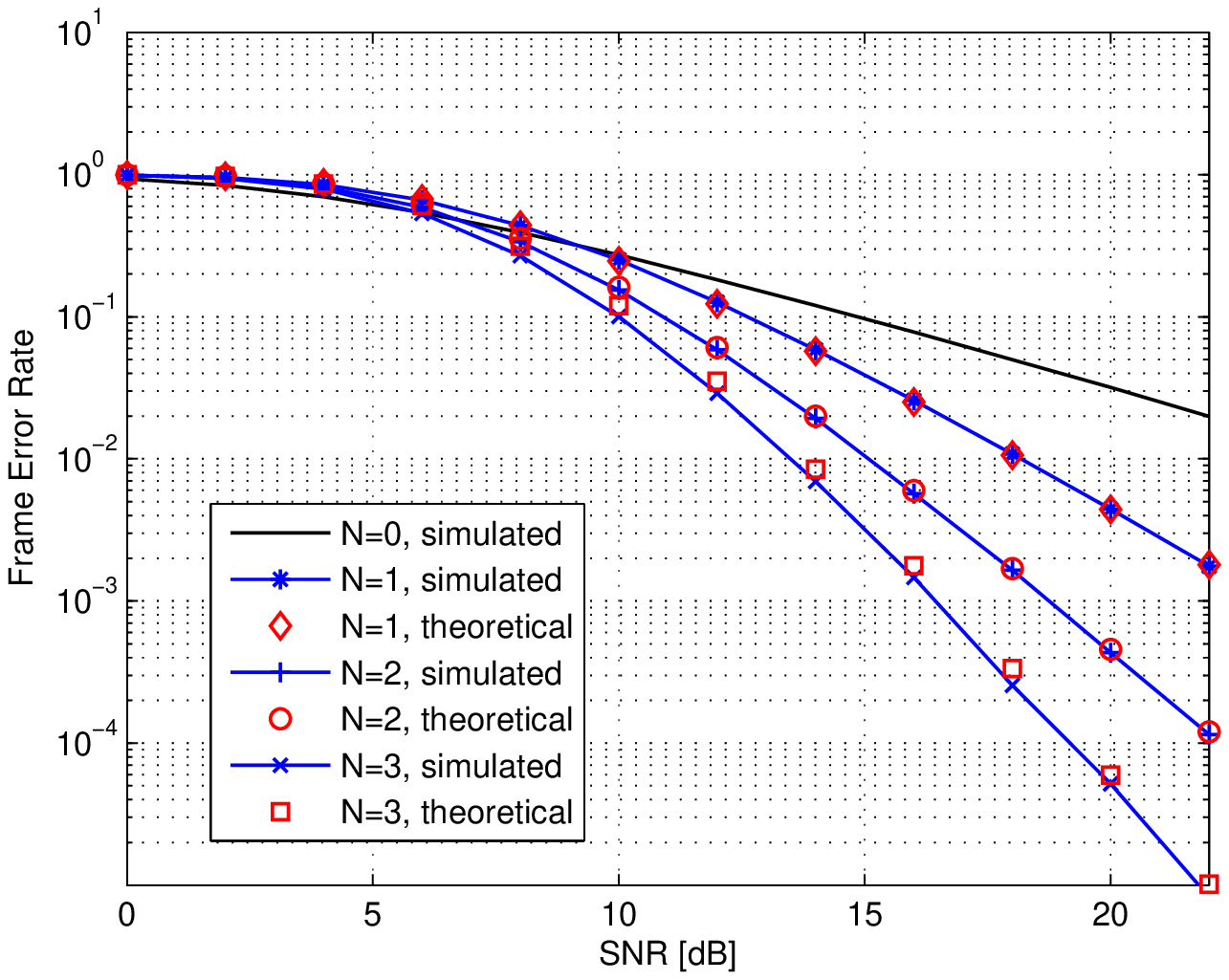}
\caption{Closed-form and simulated FER results of the SCHCN scheme under a total transmit power constraint ($N=1,2,3$, $N_c =1$) and non-cooperative system ($N=0$), $\Omega_0=\Omega_{1i}=\Omega_{2i}=1$ and $L=100$.
}  \label{fig:Theory_Simulation_HF_123Relay1_ClosedForm}
\end{figure}

%%****************************************************************
%% Simulation
\begin{figure}[!t]%[htbp]
\centering
%\graphicspath{{}}
\includegraphics[width=0.7\textwidth]{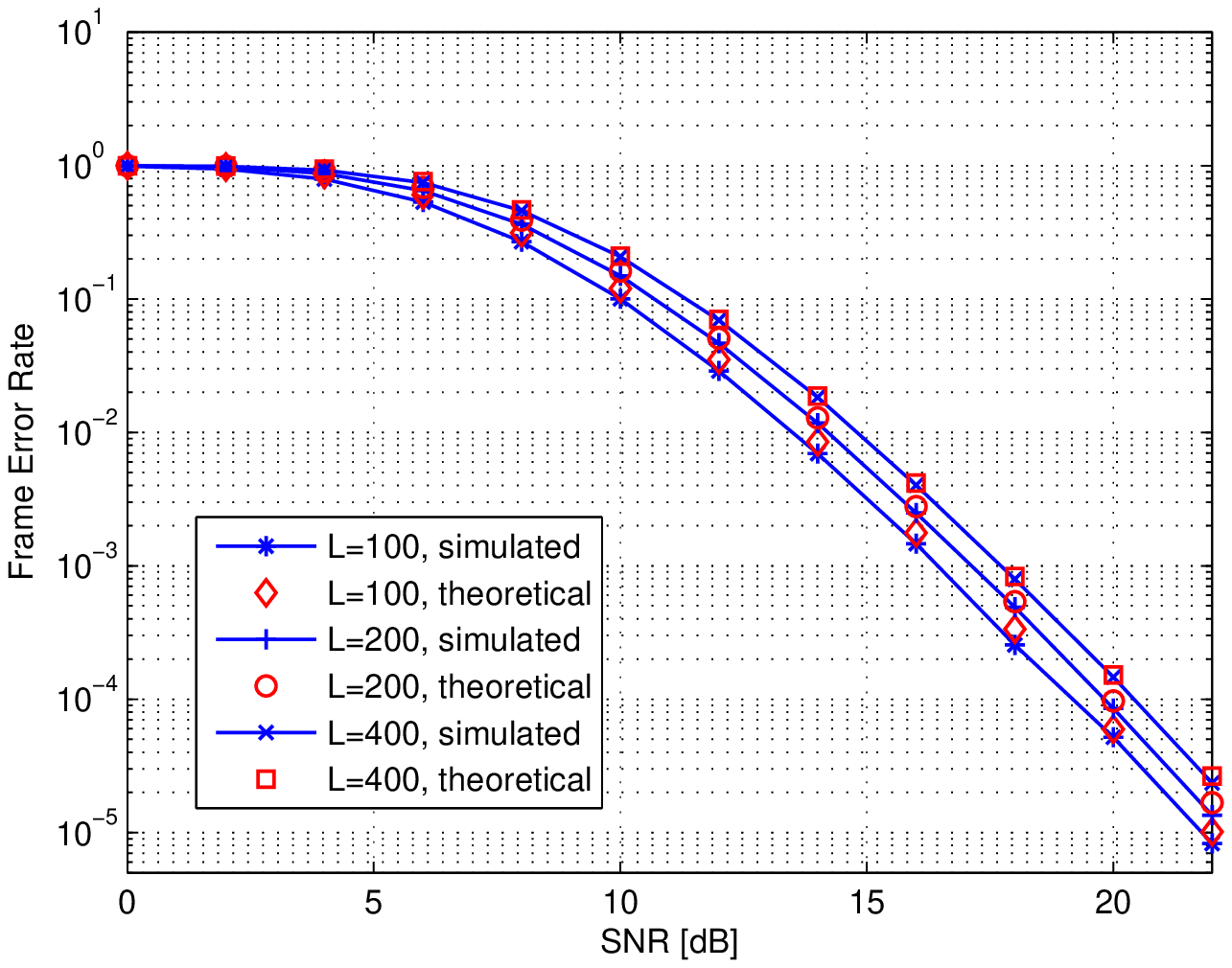}
\caption{Closed-form and simulated FER results of the SCHCN scheme under a total transmit power constraint: $N=3$, $N_c =1$, $\Omega_0=\Omega_{1i}=\Omega_{2i}=1$ and $L=100,200, 400$.
}  \label{fig:Theory_Simulation_HF_3Relay1_diff_L_ClosedForm}
\end{figure}

\end{document}